\documentclass[12pt,a4paper,twoside]{amsart} 
\usepackage{latexsym}
\usepackage{amsmath}
\usepackage{amssymb}
\usepackage{amsfonts}
\usepackage{ifthen}
\usepackage{mathrsfs}               
\usepackage{bbm}                    
\usepackage{bbold}                  
\usepackage{textcomp}               
\usepackage{amstext}
\usepackage{enumerate}
\usepackage{amstext}
\newcommand{\CC}{\mathbb{C}}

\newcommand{\NN}{\mathbb{N}}

\newcommand{\RR}{\mathbb{R}}

\newcommand{\ZZ}{\mathbb{Z}}
\newcommand{\supp}{\mathrm{supp}}
\newcommand{\const}{\mathrm{const}}
\newcommand{\dist}{\mathrm{dist}}
\newcommand{\Ran}{\mathrm{Ran}}

\newcommand{\sgn}{\mathrm{sgn}}

\newcommand{\el}{\mathrm{el}}
\newcommand{\id}{\mathbbm{1}}
\newcommand{\klg}{\leqslant} 
\newcommand{\grg}{\geqslant}          
\newcommand{\ve}{\varepsilon}
\newcommand{\vp}{\varphi}
\newcommand{\vk}{\varkappa}
\newcommand{\vr}{\varrho}

\newcommand{\wt}[1]{\widetilde{#1}}
    
\newcommand{\SPn}[2]{\langle \,#1\,|\,#2\, \rangle} 
\newcommand{\SPb}[2]{\big\langle \,#1\,\big|\,#2\, \big\rangle} 
\newcommand{\SPB}[2]{\Big\langle \,#1\,\Big|\,#2\, \Big\rangle}
\newcommand{\ol}[1]{\overline{#1}} 
\newcommand{\mr}[1]{\mathring{#1}} 
\newcommand{\bigO}{\mathcal{O}}    
\newcommand{\V}[1]{\mathbf{#1}}
\newcommand{\valpha}{\mbox{\boldmath$\alpha$}}

\newcommand{\vsigma}{\mbox{\boldmath$\sigma$}}

\newcommand{\veps}{\mbox{\boldmath$\varepsilon$}}
\newcommand{\LO}{\mathscr{L}}      

\newcommand{\HP}{\mathscr{K}}
\newcommand{\HR}{\mathscr{H}}
\newcommand{\HRp}{\mathscr{H}^+_{\mathbf{A}}}
\newcommand{\Fock}{\mathscr{F}_{\mathrm{b}}}
\newcommand{\core}{\mathscr{D}_0}
\newcommand{\dom}{\mathcal{D}}
\newcommand{\form}{\mathcal{Q}}
\newcommand{\spec}{\mathrm{\sigma}}

\newcommand{\PAm}{P^-_\mathbf{A}}               
\newcommand{\PA}{P^+_{\mathbf{A}}}

\newcommand{\PAF}{P^F_{\mathbf{A}}}
\newcommand{\PAmF}{P^{-,F}_{\mathbf{A}}}
\newcommand{\PAsharp}{P^{\sharp}_{\mathbf{A}}}
\newcommand{\PApm}{P^{\pm}_{\mathbf{A}}}
\newcommand{\PO}{P^+_{\mathbf{0}}}
\newcommand{\POpm}{P^{\pm}_{\mathbf{0}}}
\newcommand{\POm}{P^-_{\mathbf{0}}}
\newcommand{\DA}{D_{\mathbf{A}}}                
 
\newcommand{\DO}{D_{\mathbf{0}}}
\newcommand{\D}[1]{D_{#1}}
\newcommand{\R}[2]{R_{#1}(#2)}                  
\newcommand{\RA}[1]{R_{\mathbf{A}}(#1)}

\newcommand{\RO}[1]{R_{\mathbf{0}}(#1)}
\newcommand{\Hf}{H_f}                           
\newcommand{\HT}{\wt{H}_f}                      
\newcommand{\NPO}[1]{H_{#1}^{\mathrm{np}}}      
\newcommand{\PF}[1]{H_{#1}^{\mathrm{PF}}}       
\newcommand{\Xnp}[1]{X^{\mathrm{np}}_{#1}}      
\newcommand{\Ynp}[1]{Y^{\mathrm{np}}_{#1}}
\newcommand{\YnpF}[1]{Y^{\mathrm{np},F}_{#1}}
\newcommand{\XPF}[1]{X^{\mathrm{PF}}_{#1}}
\newcommand{\YPF}[1]{Y^{\mathrm{PF}}_{#1}}
\newcommand{\gV}{\frac{\gamma}{|\V{x}|}}        
\newcommand{\tgV}{\tfrac{\gamma}{|\V{x}|}}
\newcommand{\BR}[1]{B^{\mathrm{el}}_{#1}}
\newcommand{\ad}{a^\dagger}                     
\newcommand{\tipo}{i\mbox{\boldmath$\alpha$}\cdot(\nabla\chi+\chi\,\nabla F)}
\newcommand{\gcnp}{\gamma_{\mathrm{c}}^{\mathrm{np}}} %
\newcommand{\gcPF}{\gamma_{\mathrm{c}}^{\mathrm{PF}}} 
\newcommand{\UV}{\Lambda}
\newcommand{\Thnp}{\Sigma_{\mathrm{np}}}
\newcommand{\ThPF}{\Sigma_{\mathrm{PF}}}
\newcommand{\cA}{\mathcal{A}}
\newcommand{\cO}{\mathcal{O}}

\newcommand{\cT}{\mathcal{T}}

\newcommand{\cK}{\mathcal{K}}

\newcommand{\sC}{\mathscr{C}}

\newcommand{\sR}{\mathscr{R}}

\renewcommand{\Im}{\mathrm{Im}\,}
\renewcommand{\Re}{\mathrm{Re}\,}
\newtheorem{theorem}{Theorem}[section]
\newtheorem{lemma}[theorem]{Lemma}

\newtheorem{corollary}[theorem]{Corollary}
\newtheorem{hypothesis}{Hypothesis}
\theoremstyle{remark}
\newtheorem{remark}[theorem]{Remark}
\newtheorem*{example}{Example}
\numberwithin{equation}{section}
\title[Localization in relativistic QED]{
Exponential localization of hydrogen-like atoms in
relativistic quantum electrodynamics}
\author{Oliver Matte}
\address{Oliver Matte\\
Institut f\"ur Mathematik\\
TU Clausthal\\
Erzstra{\ss}e 1\\
D-38678 Clausthal-Zellerfeld, Germany\\
{\it On leave from:}
Mathematisches Institut\\
Ludwig-Maximilians-Universit\"at\\
Theresienstra{\ss}e 39\\
D-80333 M\"unchen, Germany.}
\email{matte@math.lmu.de}
\author{Edgardo Stockmeyer}
\address{Edgardo Stockmeyer\\
Mathematisches Institut\\
Ludwig-Maximilians - Universit\"at\\Theresienstra{\ss}e 39\\
D-80333 M\"unchen, Germany.}
\email{stock@math.lmu.de}

\date{\today}
\begin{document}

\begin{abstract}
We consider two different models of
a hydrogenic atom in a quantized electromagnetic field
that treat the electron relativistically. 
The first one is a no-pair model in the free picture,
the second one is given by the semi-relativistic
Pauli-Fierz Hamiltonian.
We prove that the  no-pair operator
is semi-bounded below and that its 
spectral subspaces corresponding to energies below the ionization
threshold are exponentially localized.
Both results hold true, 
for arbitrary values of the fine-structure constant,
$e^2$, 
and the ultra-violet cut-off, $\Lambda$, and
for all nuclear charges less than the critical
charge without radiation field, $Z_c=e^{-2}2/(2/\pi+\pi/2)$.
We obtain similar results for the semi-relativistic
Pauli-Fierz operator, 
again for all values of $e^2$ and $\Lambda$
and for nuclear charges less than $e^{-2}2/\pi$.
\end{abstract}

\maketitle


\section{Introduction}

\noindent
The existence of ground states of atoms and molecules
described in the framework of non-relativistic
quantum electrodynamics (QED) has been intensively
studied in the past ten years.
The first existence proofs have been given in
\cite{BFS1998b,BFS1999} for small values of the involved
physical parameters, namely the fine-structure constant, $e^2$,
and the ultra-violet cut-off, $\UV$. 
In \cite{GLL2001} the existence of ground states
for the Pauli-Fierz Hamiltonian has been established
for arbitrary values of $e^2$ and $\UV$ assuming
a certain binding condition which has been verified
later on in \cite{BCV2003} for helium-like atoms and
in \cite{LiebLoss2003} for an arbitrary number of electrons.
Moreover, infra-red finite algorithms and renormalization
group methods have been applied to various models of
non-relativistic QED to study their ground state energies and
projections
\cite{BFS1998a,BFS1998b,BFS1999,BCFS2003,BFP2006,BachKoenenberg2006,FGS2008}.
A question which arises naturally in this context is whether
these results still hold true when the electrons are 
described by a relativistic operator.
The aim of the present paper is to take one step forward
in this direction. We study  
two different models that seem to be natural candidates for a
mathematical analysis:
The first one is given by the following no-pair operator,
\begin{equation}\label{def-np-intro}
\PA\,\Big(\DA\,-\,\gV\,+\,\Hf\Big)\,\PA\,.
\end{equation}
Here $\DA$ is the free Dirac operator minimally coupled
to the quantized, ultra-violet cut-off vector potential, $\V{A}$.
(The symbol $\V{A}$ includes the fine-structure constant $e^2$.) 
$\gamma\grg0$ is a coupling constant,
$\Hf$ is the
radiation field energy, and $\PA$ the spectral projection
onto the positive spectral subspace of $\DA$.
The latter choice of projection
is referred to as the free picture. The no-pair operator
is thus acting on a projected Hilbert space where the
electron and photon degrees of freedom are always linked together.
The mathematical analysis of the analogue of this operator
for molecules has been initiated in \cite{LiebLoss2002}
where the stability of the second kind 
is shown under certain restrictions 
on $e^2$, $\UV$, and the nuclear charges. Moreover,
in \cite{LiebLoss2002b} the (positive) binding energy is
estimated from above.
There are numerous mathematical contributions on no-pair
models where magnetic fields
are not taken into account or treated classically; see, e.g.,
\cite{MatteStockmeyer2008b} for a list of references and
also for a different choice of the projections.
We remark that it is essential that the vector potential
is included in the projection determining the no-pair model.
For if $\PA$ is replaced by $\PO$
then the analogue of \eqref{def-np-intro} describing
$N$ interacting electrons becomes unstable as soon as $N\grg2$
\cite{GriesemerTix1999,LiebLoss2002,LSS1997}. Moreover, the operator in 
\eqref{def-np-intro} is formally gauge invariant
and this would not hold true anymore with
$\PO$ in place of $\PA$. Gauge invariance plays, however,
an important role in the proof of the existence of ground
states as it permits to derive 
bounds on the number of soft photons.
In fact, employing a mild infra-red regularization
it is possible to prove the existence of ground states
for the operator in \eqref{def-np-intro} with
$\PA$ replaced by $\PO$ 
\cite{Matte2000,Koenenberg2004}.
It seems, however, unlikely that the 
infra-red regularization can be dropped 
in this case \cite{Koenenberg2004}.

The second operator studied in this article,
the semi-relativistic Pauli-Fierz operator, is given as
\begin{equation}\label{def-PF2}
\sqrt{(\vsigma\cdot(-i\nabla+\V{A}))^2+\id}\,-\,\gV\,+\,\Hf\,,
\end{equation}
where $\vsigma$ is a vector containing the Pauli spin matrices.
For $\gamma=0$
the fiber decomposition with respect to different values
of the total momentum of
this operator has been studied recently in 
\cite{MiyaoSpohn2008}.
Furthermore, it is remarked in \cite{MiyaoSpohn2008} that
for $\gamma>0$, 
all eigenvalues of the operator in \eqref{def-PF2}
are at least doubly degenerate since it anti-commutes
with the time-reversal operator. 

Typically, proving the existence of ground states in QED
requires some information on the localization of
low-lying spectral subspaces or at least of certain
approximate
ground state eigenfunctions. 
Here localization is understood with respect
to the electronic degrees of freedom.
In this paper we establish this prerequisite
for both models mentioned above by proving that
spectral projectors corresponding to energies
below the ionization thresholds are still bounded
when multiplied with suitable exponential weight functions
acting on the electron coordinates.
These results hold true for all values
of the fine-structure constant $e^2$ and the ultra-violet cut-off $\UV$,
and for all coupling constants $\gamma$ below
the critical values without quantized fields.
That is, for $\gamma\in(0,2/(2/\pi+\pi/2))$ 
in the case of the no-pair operator \cite{EPS1996},
and for $\gamma\in(0,2/\pi)$ in the case of the
semi-relativistic
Pauli-Fierz operator.
The ionization thresholds are defined as the infima
of the spectra of the operators with $\gamma=0$.
Of course, our localization estimates are non-trivial
only if the infima of the spectra for $\gamma>0$ lie strictly below
the ionization thresholds. 
In the present paper we verify this binding condition 
for sufficiently small values of $e^2$ and/or $\Lambda$. 
In fact, this perturbative result is a straightforward consequence
of some of our technical lemmata.
We remark that
up to now it has actually not been known that the quadratic forms
of both operators treated here are semi-bounded below when $\gamma$
varies in the parameter ranges given above
and $e^2$ and $\UV$ are arbitrary.
The proof of this is our first main result.
For the semi-relativistic
Pauli-Fierz operator we prove the semi-boundedness
also in the critical case $\gamma=2/\pi$.
Moreover, 
the relation which determines the
exponential decay rates, $a>0$, of the semi-relativistic
Pauli-Fierz operator
in terms of the ionization threshold 
does not depend on $e^2$
and $\UV$ either. 
We have, however, to content ourselves with suboptimal
estimates on $a$ because of technical reasons.
In the case of the no-pair operator we find a relation 
between $a$ and the ionization threshold which
does depend on $e^2$ and $\UV$ and it seems to be difficult 
to avoid this.
In fact,
what complicates the analysis of both models 
is the non-locality of the corresponding Hamiltonians.
In this respect the no-pair operator is harder to analyze since
also the potential and radiation field energy become non-local.
In order to deal with this
we derive various estimates on commutators involving
the spectral projection $\PA$, exponential weights,
and cut-off functions.  
We already obtained similar bounds for spectral projections
of Dirac operators in classical magnetic fields in 
\cite{MatteStockmeyer2008a,MatteStockmeyer2008b}.
However, since we are now dealing with quantized fields
we additionally have to study commutators involving
the quantized field energy.

We remark that 
the ionization threshold is expected to coincide
with the energy value separating exponentially localized spectral
subspaces from non-localized ones which requires
also an upper bound on the energy of localized states.
In non-relativistic QED this picture has been established in
\cite{Griesemer2004} again for arbitrary values of $e^2$ and $\UV$.

We further remark that the existence of ground
states in a relativistic model describing
both the photons and the electrons and positrons
by quantized fields has been studied in
\cite{BDG2002,BDG2004}. To this end 
infra-red and ultra-violet cut-offs for the momenta
of all involved particles are imposed in the interaction part
of the Hamiltonian considered in \cite{BDG2002,BDG2004}. 

Finally, we would like to announce that this work will be
continued by M.~K\"onenberg and the present authors
in \cite{KMS2009a} where the existence of ground states
is established for both models treated in the present article.

\bigskip

\noindent
{\em This article is organized as follows.}
In the subsequent section we introduce the
no-pair and semi-relativistic Pauli-Fierz operators
and state our main results precisely.
Section~\ref{sec-commutators} provides various technical
ingredients, for instance commutator estimates that
describe the non-local properties of $\PA$.
In Section~\ref{sec-sb} we prove the semi-boundedness for both models
and, finally, in Section~\ref{sec-el} we prove the exponential
localization. The main text is followed by an appendix
where we derive simple perturbative estimates on the ionization thresholds
and ground state energies for small $e^2$ and/or $\UV$.


\section{Definition of the models and main results}

\noindent
In order to introduce the models treated in this article more precisely
we first fix our notation and recall some standard facts.
The state space of the 
quantized photon field is the bosonic Fock space,
$$
\Fock[\HP]\,:=\,
\bigoplus_{n=0}^\infty\Fock^{(n)}[\HP]
\,\ni\psi\,=\,(\psi^{(0)},\psi^{(1)},\psi^{(2)},\ldots\;)
\,.
$$
It is 
modeled over the one photon Hilbert space
$$
\HP\,:=\,L^2(\cA\times\ZZ_2,dk)\,,\quad
\int dk\,:=\,\sum_{\lambda\in\ZZ_2}\int_{\cA}d^3\V{k}\,.
$$
We assume that $\cA$ is $\RR^3$ or $\RR^3$ 
with a ball about the origin removed since these are the 
examples we encounter in \cite{KMS2009a}.
$k=(\V{k},\lambda)$
denotes a tupel consisting of a photon wave vector,
$\V{k}\in\cA$, and a polarization label, $\lambda\in\ZZ_2$.
Moreover, $\Fock^{(0)}[\HP]:=\CC$ and
$\Fock^{(n)}[\HP]:=L^2_s\big((\cA\times\ZZ_2)^n\big)$
is the subspace of all complex-valued, square integrable functions
on $(\cA\times\ZZ_2)^n$ that remain invariant under permutations
of the $n\in\NN$ wave vector/polarization tupels. 
As usual we denote the vacuum vector by
$
\Omega:=(1,0,0,\dots)\in\Fock[\HP]
$.
Many calculations will be performed on the 
following dense subspace of $\Fock[\HP]$,
$$
\sC_0\,:=\,\CC\oplus\bigoplus_{n\in\NN}
C_0((\cA\times\ZZ_2)^n)\cap L^2_s((\cA\times\ZZ_2)^n)\,.
\quad\textrm{(Algebraic direct sum.)}
$$
The free field energy of the photons
is the self-adjoint operator given by
\begin{align*}
\dom(\Hf)\,:=\,
\Big\{\,&(\psi^{(n)})_{n=0}^\infty\in\Fock[\HP]\::
\\
&
\sum_{n=1}^\infty\int\Big|\sum_{j=1}^n\omega(k_j)\,\psi^{(n)}(k_1,\dots,k_n)
\Big|^2\,dk_1\dots dk_n\,<\,\infty
\,\Big\}\,,
\end{align*}
and, for $\psi\in\dom(\Hf)$, 
$$
(H_f\,\psi)^{(0)}\,=\,0\,,\quad
(H_f\,\psi)^{(n)}(k_1,\dots,k_n)\,=\,
\sum_{j=1}^n\omega(k_j)\,\psi^{(n)}(k_1,\dots,k_n)\,,\quad n\in\NN\,.
$$
Here the dispersion relation 
$\cA\times\ZZ_2\ni k\mapsto\omega(k)$, 
depends only on $\V{k}$ and not on $\lambda\in\ZZ_2$.
Its precise form is not important in this paper. It could
be any positive, polynomially bounded, measurable function.
For definiteness we
assume that $0\klg\omega(k)\klg|\V{k}|$, $k\in\cA\times\ZZ_2$,
since this is sufficient to apply our results in \cite{KMS2009a}.
By symmetry and Fubini's theorem
\begin{equation}\label{Hf=dGamma}
\SPb{H_f^{1/2}\,\phi}{H_f^{1/2}\,\psi}\,=\,\int\omega(k)\,
\SPn{a(k)\,\phi}{a(k)\,\psi}\,dk\,,
\qquad \phi,\psi\in\dom(\Hf^{1/2})\,,
\end{equation}
where $a(k)$ 
annihilates a photon with
wave vector/polarization $k$, 
$$
(a(k)\,\psi)^{(n)}(k_1,\dots,k_n)\,=\,
(n+1)^{1/2}\,\psi^{(n+1)}(k,k_1,\dots,k_n)\,,\quad n\in\NN_0\,,
$$
almost everywhere,
and $a(k)\,\Omega=0$.
We further recall that the creation
and the annihilation operators
of a photon state $f\in\HP$ are given
by
\begin{align*}
(\ad(f)\,\psi)^{(n)}(k_1,\ldots,k_n)\,&=\,
n^{-1/2}\sum_{j=1}^nf(k_j)\,\psi^{(n-1)}(\ldots,k_{j-1},k_{j+1},\ldots)\,,
\quad n\in\NN\,,
\\
(a(f)\,\psi)^{(n)}(k_1,\ldots,k_n)\,&=\,
(n+1)^{1/2}\int\ol{f}(k)\,\psi^{(n+1)}(k,k_1,\ldots,k_n)\,dk\,,
\quad n\in\NN_0\,,
\end{align*}
and $(\ad(f)\,\psi)^{(0)}=0$, $a(f)\,\Omega=0$.
We define $\ad(f)$ and $a(f)$ on their 
maximal domains.
The following canonical commutation relations hold true
on $\sC_0$,
$$
[a(f)\,,\,a(g)]\,=\,[\ad(f)\,,\,\ad(g)]\,=\,0\,,\qquad
[a(f)\,,\,\ad(g)]\,=\,\SPn{f}{g}\,\id\,,
$$
where $f,g\in\HP$. Moreover, we have
$\SPn{a(f)\,\phi}{\psi}=\SPn{\phi}{\ad(f)\,\psi}$, $\phi,\psi\in\sC_0$,
and, by definition, $a(f)\,\phi=\int \ol{f}(k)\,a(k)\,\phi\,dk$, 
$\phi\in\sC_0$.

Next, we describe the interaction between four-spinors
and the photon field.
The full Hilbert space containing all electron/positron
and photon degrees of freedom is
$$
\HR\,:=\,L^2(\RR^3_{\V{x}},\CC^4)\otimes\Fock[\HP]\,.
$$
It contains the dense subspace,
$$
\core\,:=\,C_0^\infty(\RR^3_{\V{x}},\CC^4)\otimes\sC_0\,.
\quad\textrm{(Algebraic tensor product.)}
$$
We consider general form factors fulfilling
the following condition:

\begin{hypothesis}\label{hyp-G}
For every 
$k\in(\cA\setminus\{0\})\times\ZZ_2$ and $j\in\{1,2,3\}$,
$G^{(j)}(k)$
is a bounded continuously differentiable function,
$\RR^3_{\V{x}}\ni\V{x}\mapsto G^{(j)}_{\V{x}}(k)$,
such that
\begin{equation}
\label{def-d3}
2\int\omega(k)^{\ell}\,\|\V{G}(k)\|^2_\infty\,dk
\,\klg\,d_\ell^2\,,
\qquad \ell\in\{-1,0,1,2\}\,,
\end{equation}
and 
\begin{equation}\label{hyp-rotG}
2\int\omega(k)^{-1}\,\|\nabla_{\V{x}}\wedge\V{G}(k)\|^2_\infty\,dk
\,\klg\,d_{1}^2\,,
\end{equation}
for some $d_{-1},\ldots,d_2\in(0,\infty)$.
Here 
$\V{G}_{\V{x}}(k)=
\big(G^{(1)}_{\V{x}}(k),G^{(2)}_{\V{x}}(k),G^{(3)}_{\V{x}}(k)\big)$
and
$\|\V{G}(k)\|_\infty:=\sup_{\V{x}}|\V{G}_{\V{x}}(k)|$.
\end{hypothesis}

\smallskip

\noindent
Although we are interested in the specific physical situation
described in the following example we work with the more
general hypothesis above since in our future applications
we shall encounter truncated and discretized versions
of the vector potential and the field energy.
It will then be necessary to know that the results of the
present article hold uniformly in the truncation and
discretization and Hypothesis~\ref{hyp-G}
is a convenient way to handle this.

\begin{example}
In the physical models we are interested in we have
\begin{equation}\label{Gphys}
G^{e^2,\UV}_\V{x}(k)\,
:=\,e^2\,\frac{\id_{\{|\V{k}|\klg\UV\}}}{\pi\sqrt{2|\V{k}|}}
\,e^{-i\V{k}\cdot\V{x}}\,\veps(k),
\end{equation}
for $(\V{x},k)\in\RR^3\times(\RR^3\times\ZZ_2)$
with $\V{k}\not=0$.
Here energies are measured in units of $mc^2$, $m$ denoting the
rest mass of an electron and $c$ the speed of light.
Length, i.e. $\V{x}$, are measured in units of $\hslash/(mc)$,
which is the Compton wave length devided by $2\pi$.
$\hslash$ is Planck's constant divided by $2\pi$.
The photon wave vectors
$\V{k}$ are measured in units of $2\pi$ times
the inverse Compton wavelength,
$mc/\hslash$. The parameter
$\UV>0$ is an ultraviolet cut-off
and $e^2\approx1/137$ denotes Sommerfeld's
fine-structure constant which equals the square of
the elementary charge in our units.
One could equally well impose a smooth ultra-violet cut-off. 
The polarization vectors, $\veps(\V{k},\lambda)$, $\lambda\in\ZZ_2$,
are homogeneous of degree zero in $\V{k}$ such that
$\{\mr{\V{k}},\veps(\mr{\V{k}},0),\veps(\mr{\V{k}},1)\}$
is an orthonormal basis of $\RR^3$,
for every $\mr{\V{k}}\in S^2$.
This corresponds to the Coulomb gauge. 
(In particular, the vector fields
$S^2\ni\mr{\V{k}}\mapsto\veps(\mr{\V{k}},\lambda)$
cannot be continuous; see \cite{LiebLoss2004}
for more information on the choice of $\veps$).\hfill$\diamond$
\end{example}

\smallskip

\noindent
Finally, we introduce the self-adjoint Dirac matrices
$\alpha_1,\alpha_2,\alpha_3$, and $\beta$ 
that act on the four spinor components of an element from $\HR$.
They are given by
$$
\alpha_j\,:=\,
\begin{pmatrix}
0&\sigma_j\\\sigma_j&0
\end{pmatrix}\,,\quad j\in\{1,2,3\}\,,\qquad
\beta\,:=\,\alpha_0\,:=\,
\begin{pmatrix}
\id&0\\0&-\id
\end{pmatrix}\,,
$$
where $\sigma_1,\sigma_2,\sigma_3$ denote the standard Pauli matrices,
and fulfill the
Clifford algebra relations
\begin{equation}\label{Clifford}
\alpha_i\,\alpha_j\,+\,\alpha_j\,\alpha_i\,=\,2\,\delta_{ij}\,\id\,,
\qquad i,j\in\{0,1,2,3\}\,.
\end{equation}
The interaction between the electron/positron
and photon degrees of freedom is now given as
$$
\valpha\cdot\V{A}
\,\equiv\,\valpha\cdot\V{A}(\V{G})
\,:=\,\valpha\cdot\ad(\V{G})+\valpha\cdot a(\V{G})\,,
\quad
\valpha\cdot a^\sharp(\V{G})\,:=\,
\sum_{j=1}^3\,\alpha_j\,a^\sharp(G^{(j)}_\V{x})\,,
$$
where $a^\sharp$ is $a$ or $\ad$.
The following 
relative bounds are well-known
and show that $\valpha\cdot\V{A}$ is a symmetric operator
on $\dom(\Hf^{1/2})$.
(Henceforth we identify $\Hf^{1/2}\equiv\id\otimes\Hf^{1/2}$ etc.)
For every $\psi\in\dom(\Hf^{1/2})$,
\begin{eqnarray}
\|\valpha\cdot a(\V{G})\,\psi\|^2
&\klg&d_{-1}^2 \,\|H_f^{1/2}\,\psi\|^2,\label{rb-a}
\\
\|\valpha\cdot\ad(\V{G})\,\psi\|^2&\klg&
d_{-1}^2\,\|H_f^{1/2}\,\psi\|^2\,+\,d_0^2\,\|\psi\|^2.\label{rb-ad}
\end{eqnarray}
(Notice that the $C^*$-equality and \eqref{Clifford} imply
$\|\valpha\cdot \V{u}\|=|\V{u}|$, for every $\V{u}\in\RR^3$, whence
$\|\valpha\cdot \V{z}\|^2\klg2|\V{z}|^2$, for every $\V{z}\in\CC^3$.
This is why the factor $2$ appears on the left sides of \eqref{def-d3}
and \eqref{hyp-rotG}.)

In order to define the no-pair and semi-relativistic
Pauli-Fierz operators we recall that the
free Dirac operator minimally coupled to $\V{A}$ is given as
\begin{equation}\label{def-DA}
\DA\,:=\,\valpha\cdot(-i\nabla+\V{A})+\beta
\,:=\,
\sum_{j=1}^3\alpha_j\,\big(-i\partial_{x_j}+\ad(G_\V{x}^{(j)})
+a(G_\V{x}^{(j)})\big)\,+\,\beta
\,.
\end{equation}
A straightforward application of Nelson's commutator theorem shows that
$\DA$ is essentially self-adjoint
on $\core$ 
\cite{Arai2000,LiebLoss2002,MiyaoSpohn2008}.
We denote its closure starting from $\core$
again by the same symbol. 
Supersymmetry arguments \cite{Thaller1992} 
show that its spectrum is contained
in the union of two half-lines,
$$
\spec(\DA)\,\subset\,(-\infty,-1]\cup[1,\infty)\,.
$$
The no-pair operator acts in the projected Hilbert space
$$
\HRp\,:=\,\PA\,\HR\,,
$$ 
where
\begin{equation}\label{def-PA}
\PA\,:=\,\id_{[0,\infty)}(\DA)\,=\,\frac{1}{2}\,\id\,+\,
\frac{1}{2}\,\sgn(\DA)\,,\qquad \PAm\,:=\,\id-\PA\,.
\end{equation}
\`{A}-priori it is defined by
\begin{equation}\label{def-np}
\NPO{\gamma}\,\vp^+\,\equiv\,
\NPO{\gamma,\V{A}}\,\vp^+\,:=\,\PA\,\big(\DA-\tgV+\Hf\big)\,\vp^+\,,\qquad
\vp^+\in\PA\,\core\,.
\end{equation}
We remark that $\NPO{\gamma}$ is actually well-defined since
$\PA$ maps $\core$ into $\dom(|\V{x}|^{-1})\cap\dom(\Hf)$
by Lemmata~\ref{le-sgn} and~\ref{le-VPA-is-ok}(ii) below. 
Our first main result shows that the quadratic form
of $\NPO{\gamma}$ is always semi-bounded below, provided
$\gamma$ is less than the critical coupling constant
for the (electronic) Brown-Ravenhall operator,
$$
\BR{\gamma}\,=\,\PO\,(\DO-\tgV)\,\PO\,,
$$
which is
defined as a Friedrichs extension starting 
from $\PO\,C_0^\infty(\RR^3,\CC^4)$.
The critical coupling constant for $\BR{\gamma}$
has been determined in \cite{EPS1996} and its value is
\begin{equation}\label{def-gammac}
\gcnp\,:=\,2/(2/\pi+\pi/2)\,.
\end{equation} 
In \cite{Tix1998} it is shown that the 
Brown-Ravenhall operator
is strictly positive,
\begin{equation}\label{Tix}
\BR{\gamma}\,\grg\,(1-\gamma)\,\PO\,,\qquad \gamma\in[0,\gcnp]\,.
\end{equation}

\begin{theorem}\label{thm-sb-np}
Assume that $\V{G}$ fulfills Hypothesis~\ref{hyp-G}.
Then there is a constant, $c\in(0,\infty)$, such that,
for all 
$\gamma\in[0,\gcnp)$, $\delta>0$, 
$\rho\in(0,1-\gamma/\gcnp)$,
and $\vp^+\in\PA\,\core$, $\|\vp^+\|=1$,
\begin{align}\nonumber
\SPn{\vp^+}{(\DA-\tgV+\delta\,\Hf)\,&\vp^+}\,\grg\,\frac{1}{1+\rho}
\SPn{\vp^+}{\BR{(1+\rho)\gamma}\,\vp^+}\,+\,\SPn{\vp^+}{\POm\,|\DO|\,\vp^+}
\\
&-\,c\,(\delta+\delta^{-2})\,\label{eq-sb-np}
\big(d_1^2+d_0^2+d_{-1}^2+(d_0^2+d_{-1}^2)^2/\rho^3\big)
\,.
\end{align}
In particular, by the KLMN-theorem
$\NPO{\gamma}$ has a distinguished self-adjoint
extension -- henceforth denoted by the same symbol
-- such that $\form(\NPO{\gamma})=\form(\NPO{0})$.
Moreover, $\PA\,\core$ is a form core for $\NPO{\gamma}$.
\end{theorem}

\begin{proof}
This theorem is proved in Subsection~\ref{ssec-sb-np}.
\end{proof}

\smallskip

\noindent
On account of Tix' inequality \eqref{Tix} and Theorem~\ref{thm-sb-np}
we know that, for all $\gamma\in(0,\gcnp)$ and $\delta\in(0,1)$,
we find constants $c(\gamma), C(\gamma,\delta,d_{-1},d_0,d_1)\in(0,\infty)$
such that
\begin{equation}\label{verena}
\PA\,(\DA-\tgV+\delta\,\Hf)\,\PA
\,\grg\,c(\gamma)\,|\DO|\,-\,C(\gamma,\delta,d_{-1},d_0,d_1)\,,
\end{equation}
in the sense of quadratic forms on $\form(\NPO{\gamma})$.
We shall employ this bound in \cite{KMS2009a}.

In the sequel we denote the ionization threshold of $\NPO{\gamma}$ by
\begin{equation}\label{Th-np}
\Thnp\,\equiv\,\Thnp(\V{G})
\,:=\,\inf\big\{\,\SPn{\vp^+}{\NPO{0}\,\vp^+}\,:\;
\vp^+\in\PA\,\core\,,\;\|\vp^+\|=1\,\,\big\}\,.
\end{equation}
We denote the length of an interval $I\subset\RR$ by $|I|$.

\begin{theorem}\label{thm-el-np}
There exist constants, $k_1,k_2,k_3,k_4\in(0,\infty)$ and,
for all
$\V{G}$ fulfilling Hypothesis~\ref{hyp-G}
and all $\gamma\in(0,\gcnp)$,  we find some $E\in(0,\infty)$,
$E\equiv E(\gamma,d_{-1},d_0,d_1)\to0$, $d_i\to0$, $i\in\{-1,0,1\}$,
such that the following holds true:
Let $I\subset(-\infty,\Thnp)$ be some compact interval
and let $a>0$ satisfy $a\klg k_1(\gcnp-\gamma)/(\gcnp+\gamma)$
and $\ve:=1-\max I/(\Thnp+E)-k_2\,a^2>0$. Then
$\Ran(\id_I(\NPO{\gamma}))\subset\dom(e^{a|\V{x}|})$ and
\begin{equation}\label{eq-el-np}
\big\|\,e^{a|\V{x}|}\,\id_I(\NPO{\gamma})\,\big\|_{
\LO(\HRp,\HR)}\,\klg\,(k_3/\ve^2)\,(1+|I|)\,e^{k_4\,a\,(\Thnp+E)/\ve}\,.
\end{equation}
\end{theorem}

\begin{proof}
This theorem is proved at the end of Subsection~\ref{ssec-ed-np}.
\end{proof}

\smallskip

\noindent
Notice that the exponential decay rates $a$ in Theorem~\ref{thm-el-np}
depend on the numbers $d_i$ and $\Thnp$ but not on the particular
shape of the form factor $\V{G}$. This information on $a$
is sufficient in order to prove the existence of ground states.
We remark that
in \cite{MatteStockmeyer2008b} the present authors prove
that an eigenfunction for an eigenvalue $\lambda<1$ of
a one-particle no-pair operator in a {\em classical} magnetic
field decays with an exponential rate $a<\sqrt{1-\lambda^2}$,
for $\lambda\in[0,1)$, and $a<1$, for $\lambda<0$. This is
the behaviour known from the square-root, or, Chandrasekhar
operator. The idea used there to provide
better decay rates does, however, not apply
when the quantized field energy is present.

The simple perturbative estimates of the
following remark ensure that the statement of
Theorem~\ref{thm-el-np} is non-trivial, i.e. that
$\inf\spec(\NPO{\gamma})<\Thnp$, at least for small
values of $d_{-1},d_0,d_1$. (Recall from \cite{EPS1996} that 
$\inf\spec(\BR{\gamma})<1$, for $\gamma\in(0,\gcnp]$.) 

\begin{remark}\label{rem-O(g)-np}
There is a constant, $C_\gamma\in(0,\infty)$,
depending only on $\gamma\in(0,\gcnp)$, such that, for
all $\V{G}$ fulfilling Hypothesis~\ref{hyp-G} with
$d_{-1},d_0,d_1\klg1$,
\begin{equation}\label{bd-KMS-np1}
0\,\klg\,\Thnp-1\,\klg\,C_\gamma\,(d_{-1}+d_0+d_1)
\end{equation}
and 
\begin{equation}\label{bd-KMS-np2}
\big|\,\inf\spec(\NPO{\gamma})-\inf\spec(\BR{\gamma})\,\big|
\,\klg\,C_\gamma\,(d_{-1}+d_0+d_1)\,.
\end{equation}
These bounds are derived in
Appendix~\ref{O(g)}.
\hfill$\diamond$
\end{remark}

\smallskip

\noindent
Next, we define the second operator studied in this article, the
semi-relativistic Pauli-Fierz
operator.
It acts in the whole space $\HR$
and is \`{a}-priori given as
\begin{equation}\label{def-PF}
\PF{\gamma}\,\vp\,\equiv\,
\PF{\gamma,\V{A}}\,\vp:=\,\big(|\DA|-\tgV+\Hf\big)\,\vp\,,\qquad
\vp\in\core\,.
\end{equation}
In fact, the operator defined in \eqref{def-PF} is a two-fold
copy of the one given in \eqref{def-PF2} since
$$
|\DA|\,=\, 
\begin{pmatrix}
           \cT_\V{A}&0\\0&\cT_\V{A}
\end{pmatrix}\,,
\qquad \cT_\V{A}\,:=\,\sqrt{(\vsigma\cdot(-i\nabla+\V{A}))^2+\id}\,.
$$
We prefer, however, to consider the operator defined by \eqref{def-PF} 
in order to have a unified notation.
The critical constant for the semi-relativistic Pauli-Fierz
operator is given by Kato's constant,
\begin{equation}\label{def-gcPF}
\gcPF\,:=\,2/\pi\,.
\end{equation}

\begin{theorem}\label{le-sb-PF4}
There is some $k\in(0,\infty)$ such that, for all $\delta>0$
and $\V{G}$ fulfilling Hypothesis~\ref{hyp-G},
\begin{equation}\label{gustav}
\frac{1}{4}\,\big\|\,|\V{x}|^{-1}\,\vp\,\big\|^2\,\klg\,
\big\|\,\big(|\DA|+\delta\,\Hf+(\delta^{-1}+\delta\,k^2)
\,d_1^2\big)\,\vp\,\big\|^2\,\,,
\end{equation}
for all $\vp\in\core$, and
\begin{equation}\label{gustav2}
\frac{2}{\pi}\frac{1}{|\V{x}|}\,\klg\,|\DA|+\delta\,\Hf
+(\delta^{-1}+\delta\,k^2)\,d_1^2\,,
\end{equation}
in the sense of quadratic forms on $\core$.
In particular, for all $\gamma\in[0,\gcPF]$, $\PF{\gamma}$
has a self-adjoint Friedrichs extension -- henceforth
again denoted by the same symbol.
For $\gamma\in[0,\gcPF)$, we know that 
$\form(\PF{\gamma})=\form(\PF{0})$
and $\core$ is a form core for $\PF{\gamma}$.
\end{theorem}

\begin{proof}
This theorem is proved in Subsection~\ref{ssec-sb-PF}.
\end{proof}

\smallskip

\noindent
Due to
\cite[Proposition~1.2]{MiyaoSpohn2008} we know that
$|\DA|+\Hf$ is essentially self-adjoint on
$\core$, provided $d_{-1}$, $d_0$, and $d_1$
are sufficiently small. Together with \eqref{gustav}
and the Kato-Rellich theorem this shows that
$\PF{\gamma}$
is essentially self-adjoint on $\core$ as long as $\gamma\in[0,1/2)$
and $d_{-1}$, $d_0$, and $d_1$ are small.
In \cite{KMS2009a} we shall extend this result to all
values of $d_{-1}$, $d_0$, and $d_1$.

We denote the ionization threshold of $\NPO{\gamma}$ by 
\begin{equation}
\ThPF\,\equiv\,\ThPF(\V{G})
\,:=\,\inf\big\{\,\SPn{\vp}{\PF{0}\,\vp}\,:\;
\vp\in\core\,,\;\|\vp\|=1\,\,\big\}\,.
\end{equation}

\begin{theorem}\label{thm-el-PF}
There are constants, $k_1,k_2\in(0,\infty)$, such that, for all
$\V{G}$ fulfilling Hypothesis~\ref{hyp-G} and 
$\gamma\in(0,\gcPF)$, the following holds true: 
Let $I\subset(-\infty,\ThPF)$ be some compact interval and
assume that $a\in(0,1)$ satisfies 
$\ve:=\ThPF-\max I-9a^2/(1-a^2)^2>0$.
Then $\Ran(\id_I(\PF{\gamma}))\subset\dom(e^{a|\V{x}|})$
and 
\begin{equation}\label{main-est}
\big\|\,e^{a|\V{x}|}\,\id_I(\PF{\gamma})\,\big\|\,\klg\,
(k_1/\ve^2)\,(1+|I|)\,(\ThPF+k_2\,d_1^2)\,
e^{c(\gamma)\,a\,(\ThPF+k_2d_1^2)/\ve}
\,.
\end{equation}
Here $c(\gamma)\in(0,\infty)$ depends only on $\gamma$.
\end{theorem}

\begin{proof}
This theorem is proved at the end of Subsection~\ref{ssec-ed-PF}.
\end{proof}

\smallskip

\noindent
The following remark again ensures that
$\inf\spec(\PF{\gamma})<\ThPF$, at least for
small values of $d_{-1}$, $d_0$, and $d_1$.

\begin{remark}\label{rem-O(g)-PF}
There is a constant, $C_\gamma\in(0,\infty)$,
depending only on $\gamma\in(0,\gcnp)$, such that, for
all $\V{G}$ fulfilling Hypothesis~\ref{hyp-G} with
$d_{-1},d_0,d_1\klg1$,
\begin{equation}\label{bd-KMS-PF1}
0\,\klg\,\ThPF-1\,\klg\,C_\gamma\,(d_{-1}+d_0+d_1)
\end{equation}
and 
\begin{equation}\label{bd-KMS-PF2}
\big|\,\inf\spec(\PF{\gamma})-\inf\spec(|\DO|-\tgV)\,\big|
\,\klg\,C_\gamma\,(d_{-1}+d_0+d_1)\,.
\end{equation}
These bounds  are also derived
in Appendix~\ref{O(g)}.
\hfill$\diamond$
\end{remark}


\section{Commutator estimates}
\label{sec-commutators}

\noindent
In order to study the non-local no-pair and semi-relativistic
Pauli-Fierz operators
we need some control on various commutators
and error terms that typically appear 
in their analysis.
They involve resolvents and spectral projections of $\DA$, and
multiplication operators, in particular, exponential
weights or cut-off functions. Since we are dealing
with quantized fields we also have to
study commutators of the resolvents and spectral projections
with the radiation field energy.
The aim of this section is to provide 
appropriate bounds on the corresponding operator norms.
Our estimates on the error terms involving the field energy
are based on the 
next lemma.
The following quantity appears in its statement and in
various estimates below,
\begin{equation}\label{def-deltanu}
\delta_\nu^2\,\equiv\,{\delta_\nu}(E)^2\,:=\,
8\int\,\frac{w_\nu(k,E)^2}{\omega(k)}\,
\|\V{G}(k)\|_\infty^2\,dk\,,\qquad E,\nu>0\,,
\end{equation}
where
$$
w_\nu(k,E)\,:=\,
E^{1/2-\nu}\,\big((E+\omega(k))^{\nu+1/2}\,-\,E^\nu\,(E+\omega(k))^{1/2}\big)\,.
$$
We observe that $w_{1/2}(k,E)\klg\omega(k)$ and, hence, 
\begin{equation}\label{delta-d1}
\delta_{1/2}(E)\,\klg\,2\,d_1\,,\qquad E>0\,.
\end{equation}
Moreover,
\begin{equation*}
\delta_{\nu}(E)\,\klg\,\delta_\nu(1)\,,\qquad E\grg1\,.
\end{equation*}

\begin{lemma}\label{le-tim}
Let $\nu$, 
$E>0$, and set 
\begin{equation}\label{def-HT}
\HT\,:=\,\Hf+E\,. 
\end{equation}
Then
\begin{equation}
\big\|\,[\valpha\cdot\V{A}\,,\,\HT^{-\nu}]\,\HT^{\nu}\,\big\|
\,\klg\,
\delta_\nu(E)/E^{1/2}
\,.\label{combound2}
\end{equation}
\end{lemma}

\begin{proof}
We pick $\phi,\psi\in\core$ and write
\begin{eqnarray}
\lefteqn{
\SPb{\phi}{\nonumber
[\valpha\cdot\V{A}\,,\,\HT^{-\nu}]\,\HT^{\nu}
\,\psi}
}
\\
&=&
\SPb{\phi}{
[\valpha\cdot a(\V{G})\,,\,\HT^{-\nu}]\,\HT^{\nu}
\,\psi}
\,-\,\SPb{[\valpha\cdot a(\V{G})\,,\,\HT^{-\nu}]\,\phi}{
\HT^{\nu}
\,\psi}\,.\label{heidi1}
\end{eqnarray} 
By definition of $a(k)$ and $\Hf$ we have the 
pull-through formula
$a(k)\,\theta(\Hf)\,\psi=\theta(\Hf+\omega(k))\,a(k)\,\psi$,
for almost every $k$ and every Borel function $\theta$ on $\RR$,
which leads to
\begin{eqnarray*}
\lefteqn{
\big[\,a(k)\,,\,\HT^{-\nu}\big]\,\HT^{\nu}\,\psi
}
\\
&=&
\big\{\big((\HT+\omega(k))^{-\nu}-\HT^{-\nu}\big)\,(\HT+\omega(k))^{\nu+1/2}
\big\}\,
a(k)\,\HT^{-1/2}\,\psi\,.
\end{eqnarray*}
We denote the 
operator $\{\cdots\}$ by $F(k)$. 
Then $F(k)$ is bounded and
\begin{eqnarray*}
\|F(k)\|&\klg&
\int_0^1\sup_{t\grg0}\Big|\,\frac{d}{ds}\,
\frac{(t+E+\omega(k))^{\nu+1/2}}{(t+E+s\,\omega(k))^{\nu}}\Big|\,ds
\\
&=&
-\int_0^1\frac{d}{ds}\,
\frac{(E+\omega(k))^{\nu+1/2}}{(E+s\,\omega(k))^{\nu}}\,ds
\,=\,
w_\nu(k,E)/E^{1/2}
\,.
\end{eqnarray*}
Using these remarks together with the Cauchy-Schwarz inequality
and \eqref{Hf=dGamma},
we obtain
\begin{eqnarray*}
\lefteqn{
\big|\SPb{\phi}{
[\valpha\cdot a(\V{G})\,,\,\HT^{-\nu}]\,\HT^{\nu}\,\psi}\big|
}
\\
&\klg&
\int\|\phi\|\,\|\valpha\cdot\V{G}(k)\|\,
\|F(k)\|\,\big\|\,a(k)\,\HT^{-1/2}\,\psi\,\big\|\,dk
\\
&\klg&\|\phi\|\,
\Big(2\int\frac{\|F(k)\|^2}{\omega(k)}
\,\|\V{G}(k)\|^2_\infty\,dk\Big)^{1/2}
\Big(\int\omega(k)\,\big\|\,a(k)\,\HT^{-1/2}\,\psi\,\big\|^2\,dk\Big)^{1/2}
\\
&\klg&
\frac{\delta_{\nu}(E)}{2E^{1/2}}\:
\|\phi\|\,\big\|\,\Hf^{1/2}\,\HT^{-1/2}\,\psi\,\big\|\,
\,.
\end{eqnarray*}
A similar argument applied to the
second term in \eqref{heidi1} yields
$$
\big|\SPb{[\valpha\cdot a(\V{G})\,,\,\HT^{-\nu}]\,\phi}{
\HT^{\nu}
\,\psi}\big|\,\klg\,\frac{\wt{\delta}_\nu(E)}{2E^{1/2}}\,
\big\|\,\Hf^{1/2}\,\HT^{-1/2}\,\phi\,\big\|\,\|\psi\|\,,
$$
where $\wt{\delta}_\nu(E)$ is defined by \eqref{def-deltanu}
with $w_\nu(k,E)$ replaced by
$$
\wt{w}_\nu(k,E)\,:=\,
E^{1/2-\nu}\,\big(E^\nu\,(E+\omega(k))^{1/2}\,
-\,E^{2\nu}\,(E+\omega)^{1/2-\nu}\big)\,.
$$
Evidently, $\wt{w}_\nu\klg w_\nu$, thus $\wt{\delta}_\nu\klg\delta_\nu$,
which concludes the proof.
\end{proof}

\smallskip

\noindent
It is a trivial but very useful observation that,
by choosing $E$ large enough, we 
can make to norm appearing in \eqref{combound2} as
small as we please. For instance,
this is exploited to ensure
that certain Neumann series converge in the 
proof of the next corollary, 
where various commutation relations are established
that are used many times in the sequel.
In the whole paper it turns out to be convenient
to replace $\Hf$ by $\HT=\Hf+E$ in order to deal with
commutators involving the radiation field energy.
Thanks to Lemma~\ref{le-tim} commutators with
inverse powers of $\HT=\Hf+E$ can always be treated as small error terms.

\begin{corollary}\label{cor-T-Xi}
Let $z\in\CC$ and $L\in\LO(L^2(\RR^3_{\V{x}},\CC^4))$ be such that
$z\in\vr(\DA)\cap\vr(\DA+L)$ (where $L\equiv L\otimes\id$)
and set
\begin{equation}\label{def-HT-RAL}
\R{\V{A},L}{z}\,:=\,(\DA+L-z)^{-1}\,,\qquad \R{\V{A}}{z}\,:=\,
\R{\V{A},0}{z}\,.
\end{equation}
Assume that $\nu,E>0$ satisfy $\delta_\nu/E^{1/2}<1/\|\R{\V{A},L}{z}\|$, 
and introduce the following operators 
(recall \eqref{def-HT}\&\eqref{combound2}),
\begin{align}
&T_\nu\,:=\,\ol{[\HT^{-\nu}\,,\,\valpha\cdot\V{A}]\,\HT^{\nu}}\,,\nonumber
\\ 
\Xi_{\nu,L}(z)\,:=\,\sum_{j=0}^\infty\{-&\R{\V{A},L}{z}\,T_\nu\}^j\,,
\label{def-THT}
\qquad\Upsilon_{\nu,L}(z)\,:=\,\sum_{j=0}^\infty\{-T_\nu^*\,\R{\V{A},L}{z}\}^j
\,.
\end{align}
Then 
\begin{equation}\label{bd-Xi}
\|T_\nu\|\,\klg\,\delta_\nu/E^{1/2},\quad
\|\Xi_{\nu,L}(z)\|\,,\|\Upsilon_{\nu,L}(z)\|\,\klg\,\big(1-\delta_\nu\,
\|\R{\V{A},L}{z}\|/E^{1/2}\big)^{-1},
\end{equation}
and
\begin{eqnarray}
\big[\,\R{\V{A},L}{z}\,,\,\HT^{-\nu}\,\big]
&=&
\R{\V{A},L}{z}\,
\ol{[\HT^{-\nu}\,,\,\valpha\cdot\V{A}]}
\,\R{\V{A},L}{z}\label{eva1}
\\
&=&
\R{\V{A},L}{z}\,T_\nu\,\HT^{-\nu}\,\R{\V{A},L}{z}\,,\label{eva2}
\\
\HT^{-\nu}\,\R{\V{A},L}{z}
&=&\label{eva3}
\Xi_{\nu,L}(z)
\,\R{\V{A},L}{z}\,\HT^{-\nu}\,,
\\
\R{\V{A},L}{z}\,\HT^{-\nu}
&=&\label{eva3b}
\HT^{-\nu}\,\R{\V{A},L}{z}\,\Upsilon_{\nu,L}(z)\,,
\\
\ol{\big[\,\R{\V{A},L}{z}\,,\,\HT^{-\nu}\,\big]\,\HT^{\nu}}
&=&\label{eva4}
\R{\V{A},L}{z}\,T_\nu\,\Xi_{\nu,L}(z)
\,\R{\V{A},L}{z}
\,.
\end{eqnarray}
In particular, $\R{\V{A},L}{z}$ maps $\dom(\id\otimes\Hf^\nu)$
into itself.
\end{corollary}

\begin{proof}
First, we remark that since $(\DA-z)\,\core$ is dense in $\HR$
and since $z\in\vr(\DA+L)$ we also know that
$(\DA+L-z)\,\core$ is dense.
Next, we observe that,
for every $\psi\in\core$, we have 
$\HT^{-\nu}\,\psi\in\dom(\Hf^{1/2})\subset\dom(\valpha\cdot\V{A})$,
whence
\begin{eqnarray*}
\lefteqn{
\big[\,\R{\V{A},L}{z}\,,\,\HT^{-\nu}\,\big]\,(\DA+L-z)\,\psi
\,=\,
\R{\V{A},L}{z}\,\big[\,\HT^{-\nu}\,,\,\DA\,\big]\,\psi
}
\\
&=&\R{\V{A},L}{z}\,\big[\,\HT^{-\nu}\,,\,\valpha\cdot\V{A}\,\big]\,
\R{\V{A},L}{z}\,(\DA+L-z)\,\psi
\\
&=&\R{\V{A},L}{z}\,T_\nu\,\HT^{-\nu}\,
\R{\V{A},L}{z}\,(\DA+L-z)\,\psi
\,.
\end{eqnarray*}
Since $(\DA+L-z)\,\core$ is dense and since $T_\nu$ and 
$\ol{[\HT^{-\nu}\,,\,\valpha\cdot\V{A}]}$ are bounded 
due to Lemma~\ref{le-tim}
this implies
\eqref{eva1} and \eqref{eva2}.
Then \eqref{eva3} follows from \eqref{eva2}
and some elementary manipulations and \eqref{eva4} 
follows from \eqref{eva2} and \eqref{eva3}.
Finally, the last assertion follows from \eqref{eva3b}
(which is just the adjoint of \eqref{eva3} with $z$ and $L$
replaced by $\ol{z}$ and $L^*$
since
$\Upsilon_{\mu,L}(z)=\Xi_{\nu,L^*}(\ol{z})^*$).
\end{proof}

\smallskip

\noindent
We continue by stating some simple facts 
which are used in the proofs of the 
lemmata below:
First, we have the following representation of the sign
function of $\DA$
\cite[Lemma~VI.5.6]{Kato}, 
\begin{equation}\label{sgn}
\sgn(\DA)\,\varphi=\,\lim_{\tau\to\infty}\,
\int_{-\tau}^\tau\RA{iy}\,\varphi\,\frac{dy}{\pi}\,, \qquad 
\varphi\,\in\,\HR\,.
\end{equation}
Furthermore, since
$(-1,1)\subset\vr(\DA)$ the spectral calculus yields,
for all $y\in\RR$ and  $\kappa\in [0,1)$, 
\begin{equation}\label{ralf3}
\big\|\,|\DA|^{\kappa}\,\RA{iy}\,\big\|\,\klg\,
\frac{\id_{|y|<b(\kappa)}}{
\sqrt{1+y^2}}\:+\:
\frac{c(\kappa)\,\id_{|y|\grg b(\kappa)}}{|y|^{1-\kappa}}
\,=:\,\zeta_\kappa(y)
\,,
\end{equation} 
where $b(\kappa):=\kappa^{-1/2}(1-\kappa)^{1/2}$ ($1/0:=\infty$),
$c(\kappa):=\kappa ^{\kappa/2}(1-\kappa)^{(1-\kappa)/2}$,
and
\begin{equation}\label{def-K(kappa)}
K(\kappa)\,:=\,\int_\RR\frac{\zeta_\kappa(y)}{\sqrt{1+y^2}}\:
\frac{dy}{2\pi}\,<\,\infty\,,\qquad K(0)\,=\:\frac{1}{2}\,.
\end{equation}
Finally, to study their non-local properties we shall
conjugate or commute various operators with exponential weight
functions, $e^F\equiv e^F\otimes\id$, acting on the electron
coordinates, where
\begin{equation}\label{hyp-F<>}
F\in C^\infty(\RR^3_{\V{x}},\RR)\cap L^\infty(\RR^3_{\V{x}},\RR)\,,
\quad F(\mathbf{0})=0\,,\quad
F\grg0\;\;\textrm{or}\;\;F\klg0\,,\quad|\nabla F|\klg a\,,
\end{equation}
for some $a\in [0,1)$.
The next lemma shows that
the resolvent of $\DA$ stays bounded after conjugation
with $e^F$. 
Its statement is actually
well-known for classical magnetic fields; see, e.g.
\cite{BeGe1987}. 
The proof presented in \cite[Lemma~3.1]{MatteStockmeyer2008b}
applies, however, also to
quantized fields without any change and we refrain from
repeating it here.

\begin{lemma}\label{le-marah}
Let 
$y\in\RR$, $a\in[0,1)$,
and let
$F\in C^\infty(\RR^3_\V{x},\RR)$ have a fixed sign and
satisfy $|\nabla F|\klg a$.
Then 
$iy\in\vr(\DA+i\valpha\cdot\nabla F)$,
\begin{equation}\label{marah0}
e^F\,\RA{iy}\,e^{-F}
=(\DA+i\valpha\cdot \nabla F
+iy)^{-1}\!\!\upharpoonright_{\dom(e^{-F})}\,,
\end{equation}
and
\begin{equation}\label{marah1}
\big\|\,e^F\,\RA{iy}\,e^{-F}\,\big\|\,\klg\,
\frac{\sqrt{6}}{\sqrt{1+y^2}}\cdot
\frac{1}{1-a^2}
\,.
\end{equation}
\end{lemma}

\noindent
We define $J:[0,1)\to\RR$ by
$$
J(0)\,:=\,1\,,\qquad
J(a)\,:=\,\sqrt{6}/(1-a^2)\,, \qquad a\in (0,1)\,.
$$

\begin{lemma}\label{le-sgn}
Let $a,\kappa\in[0,1)$ and let $F$ satisfy \eqref{hyp-F<>}.
For all $\nu,E>0$ with 
$\delta_\nu \,J(a)/E^{1/2}<1$, we define
\begin{equation*}
S_{\nu}^F\,:=\,e^F\,\big[\,\sgn(\DA)\,,\,\HT^{-\nu}\,\big]\,\HT^{\nu}\,e^{-F}
\end{equation*}
on $\dom(\Hf^{\nu})$. Then
\begin{equation}\label{com-sgn}
\big\|\,|\DA|^\kappa\,S_\nu^F\,\big\|
\,\klg\,(1+a\,J(a))\,
\frac{K(\kappa)\,\delta_\nu\,J(a)/E^{1/2}}{
1-\delta_\nu\,J(a)/E^{1/2}}
\,.
\end{equation}
In particular, 
$\sgn(\DA)$ maps $\dom(\id\otimes\Hf^{\nu})$ into itself
and, if $E^{1/2}>\delta_\nu$, then
the following identities hold true on $\dom(\id\otimes\Hf^\nu)$,
\begin{eqnarray}\label{eva99}
\HT^{\nu}\,\sgn(\DA)&=&\sgn(\DA)\,\HT^{\nu}\,+\,S_\nu\,\HT^{\nu}\,,
\\
\sgn(\DA)\,\HT^{\nu}&=&\HT^{\nu}\,\sgn(\DA)\,
+\,\HT^{\nu}\,S^*_\nu\,,\label{eva100}
\end{eqnarray}
where $S_\nu:=(S_\nu^0)^*\in\LO(\HR)$.
\end{lemma}

\begin{proof}
We set $L:=i\valpha\cdot\nabla F$ so that
$e^F\,\RA{iy}\,e^{-F}=\R{\V{A},L}{iy}$, $y\in\RR$, and $\|L\|\klg a$.
Combining \eqref{eva4} with \eqref{sgn} we obtain,
for all $\phi\in\dom(|\DA|^{\kappa})$ and $\psi\in\dom(\Hf^{\nu})$,
\begin{align*}
\big|\SPb{|\DA|^{\kappa}\,&\phi}{S_{\nu}^F\,\psi}\big|
\,\klg\,
\int_\RR\big|\SPb{|\DA|^{\kappa}\,\phi}{
\R{\V{A},L}{iy}\,T_\nu\,\Xi_{\nu,L}(iy)
\,\R{\V{A},L}{iy}\,\psi}\big|\:\frac{dy}{\pi}
\\
&\klg\,\|T_\nu\|
\int_\RR\big\|\,|\DA|^{\kappa}\,\R{\V{A},L}{iy}\,\phi\,\big\|\,
\|\Xi_{\nu,L}(iy)\|\,\|\R{\V{A},L}{iy}\|\,\frac{dy}{\pi}
\;\|\phi\|\,\|\psi\|\,.
\end{align*}
Here we estimate $\|T_\nu\|$ by means of \eqref{combound2} and
we write 
$|\DA|^{\kappa}\,\R{\V{A},L}{iy}=|\DA|^{\kappa}
\,\R{\V{A}}{iy}\,(\id-L\,\R{\V{A},L}{iy})$ in order
to apply \eqref{ralf3}.
Moreover, \eqref{bd-Xi} and Lemma~\ref{le-marah} show that
$\|\Xi_{\nu,L}(iy)\|\klg(1-\delta_\nu J(a)/E^{1/2})^{-1}$, 
for all $y\in\RR$.
Altogether these remarks
yield the asserted estimate. 
Now, the following identity in $\LO(\HR)$,
$$
\sgn(\DA)\,\HT^{-\nu}\,
=\,\HT^{-\nu}\,\sgn(\DA)\,-\,\HT^{-\nu}\,(S_{\nu}^0)^*\,,
$$
shows that $\sgn(\DA)$ maps the domain of $\Hf^{\nu}$ into itself
and that \eqref{eva99} is valid.
Taking the adjoint of \eqref{eva99} and using 
$[\HT^{\nu}\,\sgn(\DA)]^*=\sgn(\DA)\,\HT^{\nu}$
(which is true since $\HT^{\nu}\,\sgn(\DA)$ is densely defined
and $\sgn(\DA)=\sgn(\DA)^{-1}\in\LO(\HR)$)
we also obtain \eqref{eva100}.
\end{proof}

\begin{lemma}\label{le-VPA-is-ok}
(i) For all $E>0$, $\nu\grg0$,
with $\delta_{\nu+1/2}/E^{1/2}<1$, 
and $\chi\in C^\infty(\RR^3_{\V{x}},[0,1])$,
the following resolvent formulas are valid,
\begin{eqnarray}
\lefteqn{
\big(\RO{iy}\,\chi-\chi\,\RA{iy}\big)\,\HT^{-\nu-1/2}\;=\nonumber
}
\\
& &
\RO{iy}\,\HT^{-\nu}\,\big\{\,\valpha\cdot(i\nabla\chi+\chi\,\V{A})
\,+\,T_\nu^*\,\chi\,\big\}\,\label{res-for-DA-D0}
\HT^{-1/2}\,\RA{iy}\,\Xi_{\nu+1/2,0}(-iy)^*,
\end{eqnarray}
and 
\begin{equation}\label{res-for-DA-D0-b}
\HT^{-1/2}\,\big(\RO{iy}\,\chi-\chi\,\RA{iy}\big)
\,=\,
\RO{iy}\,\HT^{-1/2}\,\valpha\cdot(\V{A}\,\chi-i\nabla\chi)\,\RA{iy}\,.
\end{equation}
(ii)
If $\delta_\nu(1)<\infty$, for some $\nu\grg1$, then
$\PA$ maps the subspace
$\dom(\DO\otimes\Hf^\nu)$ into itself.
\end{lemma}

\begin{proof}
(i): A short computation using \eqref{eva3} yields,
for every
$\vp\in\core$, 
\begin{eqnarray*}
\lefteqn{
\HT^{-\nu-1/2}\,\big(\RA{-iy}\,\chi-\chi\,\RO{-iy}\big)\,(\DO+iy)\,\vp
}
\\
&=&
-\Xi_{\nu+1/2,0}(-iy)\,\RA{-iy}\,\HT^{-\nu-1/2}\,
\valpha\cdot(\chi\,\V{A}+i\nabla\chi)\,\vp
\\
&=&
-\Xi_{\nu+1/2,0}(-iy)\,\RA{-iy}\,\HT^{-1/2}\,\big\{\,
\valpha\cdot(\chi\,\V{A}+i\nabla\chi)\,+\,\chi\,T_\nu\,\big\}\,\times
\\
& &\qquad\qquad\qquad\qquad\qquad\qquad
\times\,\HT^{-\nu}\,\RO{-iy}\,(\DO+iy)\,\vp
\,.
\end{eqnarray*}
Now, $(\DO-iy)\,\core$ is dense in $\HR$ and 
$\HT^{-1/2}\,\valpha\cdot\V{A}$ is bounded
due to \eqref{rb-a} and \eqref{rb-ad}.
Therefore, the previous computation 
implies an operator identity in $\LO(\HR)$
whose adjoint is \eqref{res-for-DA-D0}, and 
\eqref{res-for-DA-D0-b} is derived in a similar fashion.

(ii): For $\nu\grg1$, \eqref{res-for-DA-D0} shows that
the range of $\RA{iy}\,\HT^{1/2-\nu}$ is contained in 
$\Ran(\RO{iy}\otimes\HT^{-\nu+1})\subset\dom(\DO\otimes\Hf^{\nu-1})$, 
for every $y\in\RR$.
Moreover, we know from \eqref{def-PA} and Lemma~\ref{le-sgn} that
$\HT^{\nu-1/2}\,\PA\,\HT^{1/2-\nu}\in\LO(\HR)$.
Now, let $\vp\in\dom(\DO\otimes\Hf^\nu)$.
Then $\DA\,\vp\in\dom(\HT^{\nu-1/2})$ and
it follows that
$$
\PA\,\vp\,=\,\RA{0}\,\HT^{1/2-\nu}\,(\HT^{\nu-1/2}\,\PA\,\HT^{1/2-\nu})\,
\HT^{\nu-1/2}\,\DA\,\vp\,\in\dom(\DO\otimes\Hf^{\nu-1})\,.
$$
Furthermore, we know that $\PA\,\vp\in\dom(\id\otimes\Hf^{\nu})$
by Lemma~\ref{le-sgn}.
\end{proof}

\smallskip

\noindent
In our study of the exponential localization we shall
often encounter error terms involving the operator
\begin{equation}\label{def-cK}
\cK_{\chi,F}\,:=\,[\PA\,,\,\chi\,e^F]\,e^{-F}\,,
\end{equation}
where $\chi\in C^\infty(\RR^3_{\V{x}},[0,1])$ and $F:\RR^3_{\V{x}}\to\RR$
are functions of the electron coordinates 
and $F$ satisfies \eqref{hyp-F<>}.
The operator norm bounds derived in Lemma~\ref{compadre} below
provide the necessary control on $\cK_{\chi,F}$.

\begin{lemma}\label{compadre}
Let $a,\kappa \in [0,1)$ and let $F$ satisfy \eqref{hyp-F<>}.
Then we have,
for all $\nu,E>0$ with $E^{1/2}>\delta_\nu\,J(a)$,
and 
$\chi\in C^\infty(\RR^3,[0,1])$,
\begin{eqnarray}
\big\|\,|\DA|^\kappa\,\cK_{\chi,F}\,\big\|
&\klg&\label{peki}
K(\kappa)\,J(a)\,(a+\|\nabla\chi\|_\infty)\,,
\\\label{peki2}
\big\|\,|\DA|^\kappa\, \HT^{\mp\nu}\,\cK_{\chi,F}\,\HT^{\pm\nu}\,\big\|
&\klg&
K(\kappa)\,J(a)\, 
\frac{a+\|\nabla\chi\|_\infty}{(1-\delta_\nu\,J(a)/E^{1/2})^2}\,.
\end{eqnarray}
In particular,
\begin{equation}\label{bd-PAF}
\big\|\,e^F\,\PA\,e^{-F}\,\big\|\,\klg\,1+a\,J(a)/2\,.
\end{equation}
Let $k$ be the universal constant appearing in
Theorem~\ref{le-sb-PF4}.
There is another universal constant, $C\in(0,\infty)$,
such that, for all
$E>\max\{4J(a)^2,(1+k^2)\}\,d_1^2$,
\begin{eqnarray}
\big\|\,|\V{x}|^{-1/2}\,\cK_{\chi,F}\,\HT^{-1/2}\,\big\|
&\klg&C\,J(a)\,\label{peki-V}
\frac{a+\|\nabla\chi\|_\infty}{(1-2\,d_1\,J(a)/E^{1/2})^2}
\,.
\end{eqnarray}
\end{lemma}

\begin{proof} 
Using the notation introduced in \eqref{def-HT-RAL}
and Lemma~\ref{le-marah},
we have
\begin{equation}\label{com-RA-eF}
\big[\,\RA{iy}\,,\,\chi\,e^F\,\big]\,e^{-F}\,=\,
\RA{iy}\,M\,\R{\V{A},L}{iy}\,,
\end{equation}
where 
\begin{equation}\label{def-M-L}
M\,:=\tipo\qquad \textrm{and} \qquad L\,:=\,i\valpha\cdot\nabla F\,.
\end{equation}
By means of \eqref{Clifford} we find 
$\|M\|\klg(a+\|\nabla\chi\|_\infty)$. Moreover,
$\|\R{\V{A},L}{iy}\|\klg J(a)\,(1+y^2)^{-1/2}$ by Lemma~\ref{le-marah},
and the operator $\Xi_{\nu,L}(iy)$ given by \eqref{def-THT}
satisfies $\|\Xi_{\nu,L}(iy)\|\klg(1-\delta_\nu\,J(a)/E^{1/2})^{-1}$.
We further set $\wt{\Xi}_\nu(iy):=\sum_{\ell=0}^\infty\{-T_\nu\,\RA{iy}\}^\ell$
so that $\Xi_{\nu,0}(iy)\,\RA{iy}=\RA{iy}\,\wt{\Xi}_\nu(iy)$
and $\|\wt{\Xi}_\nu(iy)\|\klg(1-\delta_\nu/E^{1/2})^{-1}$.
On account of \eqref{sgn}
we now obtain, for all $\phi,\psi\in\core$,
\begin{eqnarray}
\lefteqn{\nonumber
\big|\SPb{|\DA|^{\kappa}\,\phi}{\HT^{-\nu}\,\cK_{\chi,F}\,
\HT^{\nu}\,\psi}\big|
}
\\
&\klg&\nonumber
\int_\RR\Big|
\SPB{|\DA|^{\kappa}\,\phi}{\HT^{-\nu}\,\RA{iy}\,M
\,\R{\V{A},L}{iy}\,\HT^{\nu}\,\psi}
\Big|\,\frac{dy}{2\pi}
\\
&=&\label{petra}
\int_\RR\Big|
\SPB{\RA{iy}\,|\DA|^{\kappa}\,\phi}{\wt{\Xi}_{\nu}(iy)\,M
\,\Xi_{\nu,L}(iy)\,\R{\V{A},L}{iy}\,\psi}
\Big|\,\frac{dy}{2\pi}
\,.
\end{eqnarray}
In the second step we used \eqref{eva3} twice and $[\HT^{-\nu},M]=0$.
Applying the various norm bounds mentioned above together with
\eqref{ralf3} we see that \eqref{peki2} holds true for the first
choice of the signs $\pm$. To obtain \eqref{peki2} with the second
choice of signs we proceed analogously applying \eqref{eva3b}
instead of \eqref{eva3}. (Notice that,
by Lemma~\ref{le-VPA-is-ok}, $\cK_{\chi,F}$ maps $\core$ into 
$\dom(\Hf^\nu)$.)
Also \eqref{peki} is proved in the same way. (Just ignore $\HT$.)
\eqref{peki-V} follows from \eqref{peki2}
and the inequality \eqref{gustav2}
from Theorem~\ref{le-sb-PF4} which is proved
independently.
We also use \eqref{delta-d1} to derive \eqref{peki-V}. 
\end{proof}

\begin{lemma}\label{le-dc}
Let  $a\in[0,1)$ and let $F$ satisfy \eqref{hyp-F<>}.
Let $\nu,E>0$ such that $\delta_\nu\,J(a)/E^{1/2}\klg1/2$.
Then
\begin{align}
\big\|\,|\DA|\,\big[\,\chi_1\,e^F\,,\,[\PA\,,\,\chi_2\,e^{-F}]\,\big]\,\big\|
&\,\klg\,
J(a)\prod_{i=1,2}
(a+\|\nabla\chi_i\|_\infty)\,,
\label{dc-DA}
\\
\big\|\,\HT^{\nu}\,\big[\,\chi_1\,e^F\,,\,[\PA\,,\,\chi_2\,
e^{-F}]\,\big]\,\HT^{-\nu}\,\big\|
&\,\klg\,
8\,J(a)\prod_{i=1,2}
(a+\|\nabla\chi_i\|_\infty)\,,
\label{dc-HT}
\\
\big\|\,\tfrac{1}{|\V{x}|}\,\big[\,\chi_1\,e^F\,,\,[\PA\,,\,\chi_2\,e^{-F}]\,\big]
\,\HT^{-1/2}\,\big\|
&\,\klg\,
8^{3/2}J(a)\prod_{i=1,2}
(a+\|\nabla\chi_i\|_\infty)
\,.
\label{dc-VC}
\end{align}
In \eqref{dc-VC} we assume that $E\grg(4d_1\,J(a))^2$.
\end{lemma}

\begin{proof}
Let $\phi,\psi\in\core$, $\|\phi\|=\|\psi\|=1$.
First, we derive a bound on
$$
I_{\phi,\psi}\,:=\,
\int_{\RR}\Big|\SPB{|\DA|\,\phi}{\HT^\nu\,
\big[\,\chi_1\,e^F\,,\,[\RA{iy}\,,\,\chi_2\,e^{-F}]\,\big]\,
\HT^{-\nu}\,\psi}\Big|\,\frac{dy}{2\pi}\,.
$$
Expanding the double commutator we get
$$
\big[\,\chi_1\,e^F\,,\,[\RA{iy}\,,\,\chi_2\,e^{-F}]\,\big]
\,=\,\eta(\chi_1,\chi_2,F\,;y)\,+\,
\eta(\chi_2,\chi_1,-F\,;y)\,,
$$
where
\begin{eqnarray*}
\lefteqn{
\eta(\chi_1,\chi_2,F\,;y)
}
\\
&:=&
\RA{iy}\,\valpha\cdot(\nabla\chi_1+\chi_1\,\nabla F)\,e^F\,\RA{iy}\,e^{-F}
\,\valpha\cdot(\nabla\chi_2-\chi_2\,\nabla F)\,\RA{iy}
\,.
\end{eqnarray*}
Writing $L:=i\valpha\cdot\nabla F$, we obtain
\begin{eqnarray}
\lefteqn{\nonumber
\int_{\RR}\Big|\SPB{|\DA|\,\phi}{\HT^\nu\,
\eta(\chi_1,\chi_2,F\,;y)\,
\HT^{-\nu}\,\psi}\Big|\,\frac{dy}{2\pi}
}
\\
&\klg&\nonumber
\int_{\RR}\Big|\SPB{\phi}{
|\DA|\,\RA{iy}\,\Upsilon_{\nu,0}(iy)\,
\valpha\cdot(\nabla\chi_1+\chi_1\,\nabla F)
\,\times
\\
& &\;\times\nonumber
\,\R{\V{A},L}{iy}\,\Upsilon_{\nu,L}(iy)\,
\valpha\cdot(\nabla\chi_2-\chi_2\,\nabla F)\,\RA{iy}\,
\Upsilon_{\nu,0}(iy)\,\psi
}\Big|\,\frac{dy}{2\pi}
\\
&\klg&\label{yvette1}
\frac{(a+\|\nabla\chi_1\|)(a+\|\nabla\chi_2\|)}{(1-\delta_\nu/E^{1/2})^2}
\cdot
\frac{J(a)}{1-\delta_\nu\,J(a)/E^{1/2}}
\int_{\RR}\frac{dy}{2\pi(1+y^2)}\,.
\end{eqnarray}
A bound analogous to \eqref{yvette1}
holds true when the roles of
$\chi_1$ and $\chi_2$ are interchanged and $F$ is
replaced by $-F$. Consequently, $I_{\phi,\psi}$ is bounded
by two times the right hand side of \eqref{yvette1}.
Altogether
this shows that \eqref{dc-DA} and \eqref{dc-HT} hold true.
(Just ignore $|\DA|$ or $\HT$, respectively, in the above argument.)
\eqref{dc-VC} follows from \eqref{dc-DA} and \eqref{dc-HT}
and the inequality
\begin{equation*}
\big\|\,|\V{x}|^{-1}\,\vp\,\big\|^2\,\klg\,
4\,\big\|\,|\DA|\,\vp\,\big\|^2\,+\,4\,\big\|\,\HT^{1/2}\,\vp\,\big\|^2\,,
\qquad \vp\in\core\,,
\end{equation*}
which is true for $E\grg d_1^2$ and derived independently
in the proof of Theorem~\ref{le-sb-PF4} given below;
see \eqref{bea1}. 
\end{proof}

\smallskip

\noindent
In what follows we set, for every $F$ satisfying \eqref{hyp-F<>},
\begin{equation}\label{def-PAF}
\PAF\,:=\,e^F\,\PA\,e^{-F},
\end{equation}
so that
\begin{equation}\label{eq-PAF}
\chi\,\PAF\,=\,\PA\,\chi-\cK_{\chi,F}\,,\qquad
\PAF\,\chi\,=\,\chi\,\PA-\cK_{\chi,-F}^*\,.
\end{equation}

\begin{corollary}\label{cor-compadre}
Let $\cO$ be $\DA$, $\frac{1}{|\V{x}|}$, $\HT$,
or any element of $\LO(\HR)$ with $\|\cO\|\klg1$.
Then there exists some universal constant
$K\in(0,\infty)$ such that,
for all $E\grg(4d_1\,J(a))^2$,
$a\in [0,1)$, 
$F$ satisfying \eqref{hyp-F<>}, 
$\chi\in C^\infty(\RR^3,[0,1])$, $\ve>0$, and $\vp\in\core\cup\PA\,\core$,
\begin{align}
\Big|\SPb{\vp&}{\PAF\,\chi\,\cO\,\chi\,\PAF\,\vp}
\,-\,\nonumber
\SPb{\vp}{\chi\,\PA\,\cO\,\PA\,\chi\,\vp}\Big|
\\
&\klg\,\ve\,\SPb{\vp}{\chi\,\PA\,|\cO|\,\PA\,\chi\,\vp}
+(1+\ve^{-1})\,K\,(a+\|\nabla\chi\|_\infty)^2\,J(a)^2\,
\big\|\,\HT^{1/2}\,\vp\,\big\|^2.\label{eq-cO}
\end{align}
Moreover, if $\cO$ is self-adjoint, then
\begin{align}
\Big|\Re\big[\,\SPb{\vp&}{\PAF\,\cO\,\PAF\,\vp}
\,-\,\nonumber
\SPb{\vp}{\PA\,\cO\,\PA\,\vp}\,\big]\,\Big|
\\
&\klg\,\label{eq-cO-a^2}
K\,a^2\,J(a)^2\,\big(\,\big\|\,\HT^{1/2}\,\vp\,\big\|^2
\,+\,\big\|\,\HT^{1/2}\,\PA\,\vp\,\big\|^2\,\big)\,.
\end{align}
If $\cO=\DA$ or $\cO\in\LO(\HR)$ then we can replace the norms
$\|\,\HT^{1/2}\,\vp\,\|^2$ and 
$\|\HT^{1/2}\,\vp\|^2+\|\HT^{1/2}\,\PA\,\vp\|^2$
on the right sides 
of \eqref{eq-cO} and \eqref{eq-cO-a^2} by $\|\vp\|^2$ and $2\|\vp\|^2$,
respectively. 
\end{corollary}

\begin{proof}
In view of \eqref{eq-PAF} and Lemma~\ref{le-VPA-is-ok}(ii)
we have the following operator
identity on $\core\cup\PA\,\core$,
\begin{eqnarray}
\lefteqn{\nonumber
\PAF\,\chi\,\cO\,\chi\,\PAF\,-\,\chi\,\PA\,\cO\,\PA\,\chi
}
\\
&=&\label{katie1}
-\cK_{\chi,-F}^*\,\cO\,\PA\,\chi\,-\,\chi\,\PA\,\cO\,\cK_{\chi,F}\,
-\,
\cK_{\chi,-F}^*\,\cO\,\cK_{\chi,F}\,.
\end{eqnarray}
Consequently,
the term on the left side of
\eqref{eq-cO} is less than or equal to
$$
\big\|\,|\cO|^{1/2}\,\PA\,\chi\,\vp\,\big\|\,\Big\{\sum_{\sharp=\pm}
\big\|\,|\cO|^{1/2}\,\cK_{\chi,\sharp F}\,\vp\,\big\|\,\Big\}\,
+\,
\prod_{\sharp=\pm}\big\|\,|\cO|^{1/2}\,\cK_{\chi,\sharp F}\,
\vp\,\big\|
\,.
$$
Therefore, \eqref{eq-cO} follows from Lemma~\ref{compadre}.

In order to derive \eqref{eq-cO-a^2} we write
$\cK_{F}:=\cK_{1,F}$ and infer from \eqref{katie1}
that
\begin{align}
\nonumber
\Re\big[\,\SPb{\vp}{&\PAF\,\cO\,\PAF\,\vp}
\,-\,
\SPb{\vp}{\PA\,\cO\,\PA\,\vp}\,\big]
\\
&=\label{katie2}
-\,
\Re\big[\SPb{\vp}{\PA\,\cO\,(\cK_{F}+\cK_{-F})\,\vp}\,\big]
\,-\,\Re\big[\,
\SPb{\vp}{\cK_{\chi,-F}^*\,\cO\,\cK_{\chi,F}\,\vp}\,\big]\,,
\end{align}
where
$$
\cK_{F}+\cK_{-F}\,=\,\big[\,[\PA\,,\,e^F]\,,\,e^{-F}\,\big]\,.
$$
Therefore, \eqref{eq-cO-a^2} follows from Lemma~\ref{le-dc}
applied to the first term in \eqref{katie2}
and Lemma~\ref{compadre} applied to the second term in \eqref{katie2}.
(In the case $\cO=\HT$ we apply \eqref{dc-HT} with $\nu=1/2$.)
\end{proof}

\begin{lemma}\label{le-DA-D0}
For all $\kappa \in [0,1)$, $E>(2d_1)^2$, and
$\chi\in C^\infty(\RR^3,[0,1])$, 
\begin{eqnarray}
\big\|\,|\DO|^\kappa(\PO\chi-\chi\PA)\HT^{-1/2}\big\|
&\klg&K(\kappa)\, \frac{\|\nabla\chi\|_\infty
+(d_0^2+2d_{-1}^2)^{1/2}}{1-2d_1/E^{1/2}},\label{peki1}
\\
\label{peki1b}
\big\|\,|\DA|^\kappa(\PA\chi-\chi\PO)\HT^{-1/2}\big\|
&\klg&K(\kappa) \big(\|\nabla\chi\|_\infty
+(d_0^2+2d_{-1}^2)^{1/2}\big).
\end{eqnarray}
\end{lemma}

\begin{proof}
Combining \eqref{res-for-DA-D0}
 with \eqref{sgn} we find, for $\phi,\psi\in\core$,
\begin{eqnarray*}
\lefteqn{\nonumber
\big|\SPb{|\DO|^\kappa\,\phi}{(\PO\,\chi-\chi\,\PA)\,\HT^{-1/2}\,\psi}\big|
}
\\
&=&\nonumber
\int_\RR\Big|\SPB{|\DO|^\kappa\,\phi}{
\RO{iy}\,\valpha\cdot(i\nabla\chi-\chi\,\V{A})
\,\HT^{-1/2}\,\RA{iy}\,\Upsilon_{1/2,0}(iy)
\,\psi}\Big|
\,\frac{d\eta}{2\pi}
\\
&\klg&\nonumber
\int_\RR\big\|\,|\DO|^\kappa\,\RO{iy}\,\big\|\,\|\phi\|\,
\big(\|\nabla\chi\|_\infty+\|\,\valpha\cdot\V{A}\,\HT^{-1/2}\,\|\big)\,\cdot
\\
& &\qquad\qquad\qquad\qquad\qquad\qquad\cdot\;
\|\RA{iy}\|
\,\|\Upsilon_{1/2,0}(iy)\|
\,\|\psi\|\,\frac{dy}{2\pi}\,.
\end{eqnarray*}
On account of \eqref{ralf3} and 
$\|\Upsilon_{1/2,0}(iy)\|\klg(1-2d_1/E^{1/2})^{-1}$
this implies \eqref{peki1}. The bound \eqref{peki1b} is proved
analogously by interchanging the roles of $\DO$
and $\DA$ and using the adjoint of \eqref{res-for-DA-D0-b}. 
\end{proof}

\begin{corollary}\label{cor-PA-PO}
For all $\ve>0$,
$\chi\in C_0^\infty(\RR^3,[0,1])$, and $\vp^+\in\PA\,\core$,
\begin{equation}
\big\|\,|\D{0}|^{1/2}\,\POm\,\chi\,\vp^+\,\big\|^2
\,\klg\,\label{clelia1}
\ve\,\big\|\,\HT^{1/2}\,\chi\,\vp^+\,\big\|^2
\,+\,\frac{c^4}{4\ve^2}\,\big\|\,\POm\,\chi\,\vp^+\,\big\|^2\,,
\end{equation}
where 
$c$ denotes the right hand side of \eqref{peki1} with $\kappa=3/4$.
Moreover, we have, for $\ve,\tau>0$,
\begin{eqnarray}
\lefteqn{\nonumber
\big|\SPb{\PO\,\chi\,\vp^+}{\tfrac{1}{|\V{x}|}\,\POm\,\chi\,\vp^+}\big|
}
\\
&\klg&\label{clelia2}
\tau\,\big\|\,|\D{0}|^{1/2}\,\PO\,\chi\,\vp^+\,\big\|^2
\,+\,
\ve\,\big\|\,\HT^{1/2}\,\vp^+\,\big\|^2
\,+\,\frac{c^4\pi^6}{2^{11}\ve^2\tau^3}\,
\big\|\,\POm\,\chi\,\vp^+\,\big\|^2\,.
\end{eqnarray}

\end{corollary}

\begin{proof}
Using \eqref{peki1} (which is certainly valid also with
$\PA$ replaced by $\PAm$), we first observe that
$$
\POm\,\chi\,\vp^+\,=\,
(\POm\,\chi-\chi\,\PAm)\,\vp^+\,\in\,\dom(|\D{0}|^{3/4})\,.
$$
This permits to get
\begin{eqnarray*}
\lefteqn{
\big\|\,|\D{0}|^{1/2}\,\POm\,\chi\,\vp^+\,\big\|^2
}
\\
&\klg&
\big\|\,|\D{0}|^{1/4}\,\POm\,\chi\,\vp^+\,\big\|\,
\big\|\,|\D{0}|^{3/4}\,(\POm\,\chi-\chi\,\PA)\,\vp^+\,\big\|
\\
&\klg&
\big\|\,|\D{0}|^{1/4}\,\POm\,\chi\,\vp^+\,\big\|\,
c\,\big\|\,\HT^{1/2}\,\vp^+\,\big\|
\\
&\klg&
\frac{c^2}{2\ve}\,
\SPb{\POm\,\chi\,\vp^+}{|\D{0}|^{1/2}\,\POm\,\chi\,\vp^+}
\,+\,\frac{\ve}{2}\,\big\|\,\HT^{1/2}\,\vp^+\,\big\|^2
\\
&\klg&
\frac{1}{2}\,\big\|\,|\D{0}|^{1/2}\,\POm\,\chi\,\vp^+\,\big\|^2
\,+\,\frac{c^4}{8\ve^2}\,\big\|\,\POm\,\chi\,\vp^+\,\big\|^2
\,+\,\frac{\ve}{2}\,\big\|\,\HT^{1/2}\,\vp^+\,\big\|^2\,,
\end{eqnarray*}
which implies \eqref{clelia1}.
The bound \eqref{clelia2} follows from \eqref{clelia1}
and Kato's inequality, $|\V{x}|^{-1}\klg(\pi/2)|\D{0}|$.
\end{proof}


\section{Semi-boundedness}\label{sec-sb}

\noindent
In the following two subsections we prove Theorems~\ref{thm-sb-np}
and~\ref{le-sb-PF4} which state that the no-pair and
relativistic Pauli-Fierz operators are semi-bounded
provided the coupling constant in front of the Coulomb potential
stays below the critical values $\gcnp=2/(2/\pi+\pi/2)$
and $\gcPF=2/\pi$, respectively.


\subsection{The no-pair operator: Semi-boundedness}\label{ssec-sb-np}

\begin{proof}[Proof of Theorem~\ref{thm-sb-np}]
We pick some $\rho\in(0,1-\gamma/\gcnp)$ and set
$\gamma_\rho:=(1+\rho)\,\gamma$.
By virtue of Lemma~\ref{le-VPA-is-ok}(ii) we have
$\PO\,\vp^+\in\dom(\DO\otimes\Hf^{1/2})$,
whence
\begin{eqnarray}
\SPb{\vp^+}{(\DA-\tgV)\,\vp^+}
&=&\nonumber
\frac{\gamma}{\gamma_\rho}\,
\SPb{\vp^+}{\PO\,(\DO-\tfrac{\gamma_\rho}{|\V{x}|})\,\PO\,\vp^+}
\\
& &
\;+\;\label{friedhelm1a}
(1-\gamma/\gamma_\rho)\,\SPb{\vp^+}{\PO\,\DO\,\vp^+}
\\
& &
\;+\;\nonumber
\SPb{\vp^+}{\valpha\cdot \V{A}\,\vp^+}
\,
\\
& &
\;+\;\label{friedhelm2}
\SPb{\vp^+}{\POm\,\big(\DO-\tfrac{\gamma}{|\V{x}|}\big)
\,\POm\,\vp^+}
\\
& &
\;-\;\label{friedhelm1}
2\gamma\,\Re\SPb{\vp^+}{\PO\,\tfrac{1}{|\V{x}|}\,\POm\,\vp^+}
\\
& &\label{friedhelm2000}
\;+\;\SPb{\vp^+}{\POm\,|\DO|\,\vp^+}\,
-\,\big\|\,|\DO|^{1/2}\,\POm\,\vp^+\,\big\|^2
\,.
\end{eqnarray}
We employ \eqref{clelia1} with $\ve=\delta/4$ to estimate
the second term in \eqref{friedhelm2000} from below
by $-(\delta/4)\SPn{\vp^+}{\HT\,\vp^+}-(4c^4/\delta^2)\,\|\vp^+\|^2$.
Here $\HT=\Hf+E$ and we choose $E=16\,d_1^2$. Then $c^4$ is proportional
to $(d_0^2+2d_{-1}^2)^2$. 
The term in \eqref{friedhelm1}
can be estimated from below by means of \eqref{clelia2},
where we choose $E=16\,d_1^2$, $\ve=\delta/(8\gamma)$,
and $\tau=(1-\gamma/\gamma_\rho)/(2\gamma)=\rho/(2\gamma[1+\rho])$.
With this choice of $\tau$ the portion of the kinetic energy
in \eqref{friedhelm1a} compensates for the contribution
coming from the first term on the right side in \eqref{clelia2}.
By Kato's inequality the term in 
\eqref{friedhelm2} is bounded from below by
$-(1+\pi\gamma/2)\,\|\,|\DO|^{1/2}\,\POm\,\vp^+\|^2$,
which we estimate further by means of 
\eqref{clelia1} with $E=16\,d_1^2$ and
$\ve=\delta/(4+2\pi\gamma)$.
Combining these remarks with
$\valpha\cdot\V{A}\grg-(\delta/4)\,\Hf-4d_{-1}^2/\delta$
(due to \eqref{rb-a}),
we arrive at \eqref{eq-sb-np}.
\end{proof}


\subsection{The semi-relativistic Pauli-Fierz operator: Semi-boundedness}
\label{ssec-sb-PF}

\begin{lemma}\label{le-clara}
There is a constant $k\in(0,\infty)$ such that,
for all $E\grg k^2\,d_1^2$ and all
$\phi\in\core$,
\begin{equation}\label{comm-posit}
\Re\SPb{|\DA|\phi}{\HT\,\phi}\,\grg\,
(1-k\,d_1
\,E^{-1/2})\,\big\|\,|\DA|^{1/2}\,\wt{H}^{1/2}_f\,\phi\,\big\|^2
\,.
\end{equation}
\end{lemma}

\begin{proof}
Let $\phi\in\core$ and set $\psi:=\HT^{1/2}\,\phi$.
Using \eqref{def-THT} and \eqref{eva99}, we have
\begin{eqnarray*}
\lefteqn{
\Re \SPb{\DA\,\HT^{-1/2}\,\psi}{\sgn(\DA)\,\HT^{1/2}\,\psi}
}
\\
&=&
\Re \SPb{(\DA-T^*_{1/2})\,\psi}{\HT^{-1/2}\,\sgn(\DA)\,\HT^{1/2}\,\psi}
\\
&=&
\Re \SPb{(\DA-T^*_{1/2})\,\psi}{\big(\sgn(\DA)-S_{1/2}\big)\,\psi}
\\
&\grg&
\SPn{|\DA|\,\psi}{\psi}-\big\|\,|\DA|^{1/2}\,\psi\,\big\|\,
\big\|\,|\DA|^{1/2}\,S_{1/2}\,\psi\,\big\|
\\
& &\qquad\qquad\qquad\qquad
-\,\|T_{1/2}\|\,(1+\|S_{1/2}\|)\,\|\psi\|^2.
\end{eqnarray*}
Together with \eqref{delta-d1}, \eqref{bd-Xi}, and \eqref{com-sgn} this gives
the asserted estimate.
\end{proof}

\bigskip

\begin{proof}[Proof of Theorem~\ref{le-sb-PF4}]
We pick some $\delta>0$ and choose
$E=(\delta^{-2}+k^2)\,d_1^2$,
where $k$ is
the constant appearing in Lemma~\ref{le-clara}.
To start with we recall that
\begin{equation}\label{ttt1}
\DA^2\,\phi\,=\,\big(-i\nabla+\V{A}\big)^2\,\phi\,
+\,\V{S}\cdot\V{B}\,\phi\,+\,\phi\,,\qquad \phi\in\core\,,
\end{equation}
where 
the entries of the formal vector $\V{S}$
are $S_j=\sigma_j\otimes\id_2$ and
$\V{B}$ is the magnetic field, i.e.
$$
\V{S}\cdot\V{B}\,=\,
\V{S}\cdot\ad(\nabla_{\V{x}}\wedge\V{G})
\,+\,\V{S}\cdot a(\nabla_{\V{x}}\wedge\V{G})\,.
$$ 
A standard estimate using Hypothesis~\ref{hyp-G}
shows that, for every $\vp\in\core$, 
$$
\big|\SPb{\vp}{\V{S}\cdot\V{B}\,\vp}\big|\,\klg\,
2\,d_1\,\|\vp\|\,\big\|\,\Hf^{1/2}\,\vp\,\big\|\,
\klg\,
\delta\,\SPn{\vp}{\Hf\,\vp}\,+\,d_1^2\,\delta^{-1}\,\|\vp\|^2
\,.
$$
Using
$E\grg k^2\,d_1^2$ in the third step,
we thus obtain,
for all $\phi\in \core$, 
\begin{eqnarray}
\frac{1}{4}\,\SPb{\phi}{|\V{x}|^{-2}\,\phi}&\klg&\nonumber
  \SPb{\phi}{\big(-i\nabla+\V{A}\big)^2\,\phi}
\\
&\klg&\label{bea1}
\SPb{\DA\,\phi}{\DA\,\phi}\,
+\,\delta\,\SPb{\phi}{(\Hf+\delta^{-2}\,d_1^2)\,\phi}\,-\,\|\phi\|^2
\\
&=&\nonumber
\SPb{\DA\,\phi}{\DA\,\phi}\,+\,\delta\,\SPb{\phi}{\HT\,\phi}
\,-\,(\delta\,k^2\,d_1^2+1)\,\|\phi\|^2
\\
&\klg &\nonumber
\SPb{\DA\,\phi}{\DA\,\phi}\,+\,\SPb{\phi}{\delta^2\,\HT^2\,\phi}
+\,2\Re\,\SPb{|\DA|\,\phi}{\delta\,\HT\,\phi}
\\
& &\nonumber\qquad\qquad\,-\,(\delta\,k^2\,d_1^2+3/4)\,\|\phi\|^2
\\
&=&\nonumber
\big\|\,(|\DA|\,+\,\delta\,\HT)\,\phi\,\big\|^2
\,-\,(\delta\,k^2\,d_1^2+3/4)\,\|\phi\|^2
\,.
\end{eqnarray}
Here we also used 
a diamagnetic inequality in the first step.
The diamagnetic inequalities
used here and in the first step of \eqref{ttt3} below
are well-known at least for classical magnetic fields.
They hold true, however, also for quantized fields 
due to an argument by J.~Fr\"{o}hlich
which is presented in \cite{AHS1978} and \cite{LiebLoss2002b};
see also \cite{Hiroshima1996,Hiroshima1997}.
(The basic underlying observation is that all components
$A_i(\V{x})$ and $A_j(\V{y})$, $i,j\in\{1,2,3\}$, $\V{x},\V{y}\in\RR^3$,
of the vector potential commute and can hence be diagonalized
simultanously. In this way the problem is reduced to the
classical situation.)
Since the square root is operator monotone it follows that, 
for all $\phi\in \core$, 
\begin{equation}\label{ttt3}
\frac{2}{\pi}\,\SPb{\phi}{|\V{x}|^{-1}\,\phi}\,\klg\,
\SPb{\phi}{|-i\nabla+\V{A}|\,\phi}\,\klg
\,\SPb{\phi}{(|\DA|+\delta\,\HT)\,\phi}\,.
\end{equation}
\end{proof}


\section{Exponential localization}\label{sec-el}

\subsection{Outline of the proof}

\noindent
Our next aim is to prove the main Theorems~\ref{thm-el-np}
and~\ref{thm-el-PF} which assert that low-lying spectral subspaces
of the no-pair and semi-relativistic Pauli-Fierz operators
are exponentially localized.
We recall the general strategy of the proofs in this subsection
and apply the results to the no-pair and
semi-relativistic Pauli-Fierz operators in 
Subsections~\ref{ssec-ed-np} and~\ref{ssec-ed-PF},
respectively.
The basic idea underlying the proofs 
is essentially due to \cite{BFS1998b}
and described briefly in Lemma~\ref{le-ed-BFS}.
The technical Lemma~\ref{le-MS-Y} summarizes (and simplifies)
a part of a proof
from \cite{MatteStockmeyer2008a}. 
Occasionally, we will also benefit from some
observations made in \cite{Griesemer2004}.

The spectra of
both the no-pair and the semi-relativistic Pauli-Fierz operators
will certainly be continuous up to their minima,
at least for the physically interesting choice of the form factor.
In particular,
we cannot employ eigenvalue equations to derive exponential
decay estimates. (Of course, this would be possible if were
only interested in the exponential localization of
ground state eigenfunctions.)
According to \cite{BFS1998b}
a possibility to handle this is to smooth
out the spectral projection and to apply a suitable
integral representation of the smoothed projection.
We shall employ the following formula due to Amrein et al.
\cite[Theorem~6.1.4(b)]{ABdMG1996}
which holds for every $f\in C_0^\infty(\RR)$, $\nu\in\NN$,
and every self-adjoint operator, $X$, in some Hilbert space,
\begin{eqnarray}
f(X)&=&\sum_{\vk=0}^{\nu-1}\frac{1}{\pi\,\vk!}\nonumber
\int_{\RR}f^{(\vk)}(\lambda)\,\Im\big[i^{\vk}\,(X-\lambda-i)^{-1}\big]
\,d\lambda
\\
& &+\label{for-ABdMG}
\int_0^1\frac{t^{\nu-1}}{\pi(\nu-1)!}\int_\RR f^{(\nu)}(\lambda)\,
\Im\big[i^{\nu}\,(X-\lambda-it)^{-1}\big]\,d\lambda\,dt\,.
\end{eqnarray}
The following lemma is essentially due to \cite{BFS1998b}.

\begin{lemma}\label{le-ed-BFS}
Let $X$ and $Y$ be self-adjoint operators in $\HR$
with a common domain. Let $a>0$ and
$I\subset\RR$ be a compact interval such that
$I\subset\vr(Y)$. Assume that there exist $C,C'\in(0,\infty)$
and another compact interval $J\subset\vr(Y)$ 
such that $\mr{J}\supset I$ and that, for all $F$ 
satisfying \eqref{hyp-F<>},
\begin{equation}\label{hyp-XY}
\big\|\,e^F\,(X-Y)\,\big\|\,\klg\,C\,,\quad
\sup\limits_{(\lambda,t)\in J\times(0,1]}
\big\|\,e^F\,(Y-\lambda\pm it)^{-1}\,e^{-F}\,\big\|\,\klg\,C'\,.
\end{equation}
Then $\Ran(\id_I(X))\subset \dom(e^{a|\V{x}|})$ and
$$
\big\|\,e^{a|\V{x}|}\,\id_I(X)\,\big\|\,\klg\,c(I,J)\,C\,C'\,,
$$
where
\begin{equation}\label{c(I,J)}
c(I,J)\,=\,k\,\big(1+|J|+\dist(I,J^c)^{-1}\big)\,,
\end{equation}
for some universal constant $k\in(0,\infty)$.
\end{lemma}

\begin{proof}
We find some $f\in C_0^\infty(\RR,[0,1])$ such that
$f\equiv1$ on $I$ and $\supp(f)\subset J$.
Then $e^F\,\id_I(X)=e^F\,(f(X)-f(Y))\,\id_I(X)$,
since $J\subset \vr(Y)$. Here we can rewrite $f(X)-f(Y)$
by means of
\eqref{for-ABdMG}. 
On account of \eqref{hyp-XY} and  
the second resolvent identity
we have, for every $\lambda\in J$ and $t\in(0,1]$,
\begin{equation}\label{karl}
\big\|\,e^F\,\big\{\,(X-\lambda\pm it)^{-1}
-(Y-\lambda\pm it)^{-1}\,\big\}\,\big\|\,\klg\,C\,C'/t\,.
\end{equation}
Now, we observe that the factor $t^{\nu-1}$ in \eqref{for-ABdMG}
compensates for the $1/t$ singularity in \eqref{karl}
if we choose $\nu=2$.
Using these remarks we readily find some
$c(I,J)\in(0,\infty)$ such that, for all
$F$ 
satisfying \eqref{hyp-F<>}, we have 
$\|e^F\,\id_I(X)\|\klg c(I,J)\,C\,C'$.
By an appropriate choice of $f$ we can ensure that $c(I,J)$
has the form given in \eqref{c(I,J)}.
But then $\|e^{a|\V{x}|}\,\id_I(X)\|\klg c(I,J)\,C\,C'$
holds true also as a consequence of 
the monotone convergence theorem
applied to a suitable increasing sequence of weights $F_1,F_2,\ldots\;$,
where each $F_j$ satisfies \eqref{hyp-F<>}.
\end{proof}

\bigskip

\noindent
To verify the second condition in \eqref{hyp-XY}
the following lemma is helpful.

\begin{lemma}\label{le-MS-Y}
Let $Y$ be a positive operator in $\HR$
which 
admits $\core$ as a form core.
Set $b:=\inf\spec(Y)$ and
let $J\subset(-\infty,b)$ be some compact interval.
Let $a\in(0,1)$ and assume that, for all
$F$ satisfying \eqref{hyp-F<>},
we have $e^{\pm F}\,\form(Y)\subset\form(Y)$.
(Notice that $e^{\pm F}$ maps $\core$ into itself.)
Assume further
that there exist constants $c(a),f(a),g(a),h(a)\in[0,\infty)$
such that $c(a)<1/2$ and $b\,g(a)+h(a)<b-\max J$ and, for all
$F$ satisfying \eqref{hyp-F<>} and $\vp\in\core$,
\begin{eqnarray}
\big|\SPb{\vp}{(e^{F}\,Y\,e^{-F}-Y)\,\vp}\big|&\klg&\label{eq-Y2a}
c(a)\,\SPn{\vp}{Y\,\vp}\,+\,f(a)\,\|\vp\|^2\,,
\\ \label{eq-Y2}
\Re\SPb{\vp}{e^F\,Y\,e^{-F}\,\vp}
&\grg&(1-g(a))\,\SPb{\vp}{Y\,\vp}\,-\,h(a)\,\|\vp\|^2\,.
\end{eqnarray}
Then we have, for all 
$F$ satisfying \eqref{hyp-F<>},
\begin{equation}\label{eq-Y3}
\sup_{(\lambda,t)\in J\times(0,1]}
\big\|\,e^F\,(Y-\lambda\pm it)^{-1}\,e^{-F}\,\big\|\,\klg\,
\big(b-\max J-h(a)-b\,g(a)\big)^{-1}
\,.
\end{equation}
\end{lemma}

\begin{proof}
Since $ e^F$ is an isomorphism on $\HR$
the densely defined operators 
$e^F\,Y\,e^{-F}$ and $Y$
have the same resolvent set and 
\begin{equation}\label{def-RF}
\sR_F(z)\,:=\,
e^F\,(Y-z)^{-1}\,e^{-F}=(e^F\,Y\,e^{-F}-z)^{-1},
\qquad z\in\vr(Y)
\,.
\end{equation}
In particular, 
$e^F\,Y\,e^{-F}$ is closed because its resolvent set is not
empty.  
Since $e^{-F}$ is a
self-adjoint isomorphism  we further know that
\begin{eqnarray}
(e^F\,Y\,e^{-F})^*&=&e^{-F}\,(e^F\,Y)^*\,=\,e^{-F}\,Y\,e^F\,.\label{mona4a}
\end{eqnarray}
By assumption we have 
\begin{eqnarray}
\dom(e^{\pm F}\,Y\,e^{\mp F})\,=\,e^{\pm F}
\dom(Y)\,\subset\,\label{mona4b}
e^{\pm F}\form (Y)\,\subset\,\form(Y)\,.
\end{eqnarray} 
Condition \eqref{eq-Y2a} and $c(a)<1/2$ imply that
$\big(e^F\,Y\,e^{-F}\big)\!\!\upharpoonright_{\core}$ 
has a distinguished closed and sectorial extension which we denote by $Y_F$.
This extension is the
only closed extension having the
properties $\dom(Y_F)\subset \form(Y)$,
$\dom(Y_F^*)\subset \form(Y)$,
and $it\in\vr(Y_F)$, for all $t\in\RR$ such that $|t|$
is larger than some positive constant; see
\cite{Kato}. 
Thanks to \eqref{def-RF}, \eqref{mona4a}, and \eqref{mona4b}, 
we know
that
$e^F\,Y\,e^{-F}$ {\em is} a closed extension enjoying all
these properties, whence 
$
Y_F=e^F\,Y\,e^{-F}
$.
By virtue of \eqref{eq-Y2} 
we have, with $\delta:=b-\max J-b\,g(a)-h(a)>0$ and for all
$\lambda\in J$, $t\in(0,1]$, 
$\vp\in \core$,
\begin{equation}
  \label{mona8}
\Re\,\SPb{\vp}{(Y_F-\lambda\pm it)\,\vp}\,\grg\,
\big\{(1-g(a))\,b-\lambda-h(a)\big\}\,\|\vp\|^2\,\grg\,
\delta\,\|\vp\|^2
\,.
\end{equation}
Therefore, the numerical range of $Y_F-\lambda\pm it$
is contained in the half space $\{\zeta\in\CC:\,\Re \zeta\grg\delta\}$
\cite[Theorem~VI.1.18 and Corollary~VI.2.3]{Kato}.
Moreover, by \eqref{def-RF} the deficiency of 
$Y_F-\lambda\pm it$ is zero, 
and we may hence estimate the norm
of $(Y_F-\lambda\pm it)^{-1}$ by the inverse distance of $\lambda\pm it$
to the numerical range of $Y_F$ \cite[Theorem~V.3.2]{Kato}.
We thus obtain the estimate 
$
\|\,(Y_F-\lambda\pm it)^{-1}\,\|\klg\delta^{-1}
$,
for all
$\lambda\in J$ and $t\in(0,1]$,
which together with \eqref{def-RF} proves the lemma.
\end{proof}


\subsection{The no-pair operator: Localization}\label{ssec-ed-np}

\noindent
To begin with
we introduce a scaled partition of unity. 
Namely, we pick some $\tilde{\mu}\in C_0^\infty(\RR^3,[0,1])$
such that $\tilde{\mu}\equiv1$ on $\{|\V{x}|\klg1\}$ 
and $\tilde{\mu}\equiv0$ on $\{|\V{x}|\grg2\}$
and observe that
$\theta:=\tilde{\mu}^2+(1-\tilde{\mu})^2\grg1/2$.
Then we
set, for $R\grg1$ and $x\in\RR^3$,
$\mu_{1,R}(x):=\tilde{\mu}(x/R)/\theta^{1/2}(x/R)$,
and 
$\mu_{2,R}(x)
:=(1-\tilde{\mu}(x/R))/\theta^{1/2}(x/R)$, so that
$\mu_{1,R}^2+\mu_{2,R}^2=1$. We define
\begin{equation}
  \label{eq:20}
  e(\gamma)\,:=\, \inf\,\spec(\NPO{\gamma}+E\,\PA)\,,
\qquad \gamma\in[0,\gcnp)\,,
\end{equation}
where $\NPO{\gamma}$ is considered as an operator acting in $\HRp$.
The parameter $E>0$ is chosen sufficiently large later on.
We shall apply Lemmata~\ref{le-ed-BFS} and~\ref{le-MS-Y}
with
\begin{eqnarray}
\Xnp{\gamma}&=&\NPO{\gamma}+E\,\PA\,+\,\PAm\,\Hf\,\PAm
\,+\,e(0)\,\PAm\,,
\\
\Ynp{\gamma}&=&\,\Xnp{\gamma}\,+\,(e(0)-e(\gamma_R))\,\PA\,
\mu_{1,R}^2\,\PA\,,
\end{eqnarray}
where $\NPO{\gamma}$ is now considered as an operator acting
in $\HR$ and
\begin{equation}\label{def-gammaR}
\gamma_R\,:=\,(1+1/R)\,\gamma/(1-c/R)\,,\qquad R>c\,.
\end{equation}
Here $c\grg1$ is the constant
appearing in Lemma~\ref{boundbelow}.
$\Xnp{\gamma}$ and $\Ynp{\gamma}$ are self-adjoint
on the same domain
$\dom(\Xnp{\gamma})=\dom(\Ynp{\gamma})=
\dom(\NPO{\gamma})\cap\dom(\PAm\,\Hf\,\PAm)$
and both operators 
admit $\core$ as a form core.

The idea to define the comparison operator $\Ynp{\gamma}$
by essentially adding only a cut-off function located
in a ball about the origin to $\Xnp{\gamma}$ is borrowed
from \cite{Griesemer2004}. An obvious consequence 
of this choice (and the bound \eqref{bd-PAF})
is that the first condition in
\eqref{hyp-XY} is fulfilled.

\begin{lemma}\label{boundbelow}
Let $\gamma\in[0,\gcnp)$ and $E>\max\{1\,,\,(2d_1)^2\}$. 
Then there exist universal constants $c,c'\in[1,\infty)$ 
such that, 
for all 
$R>\max\{\,c\,,\,(\gamma+c\,\gcnp)/(\gcnp-\gamma)\,\}$  
and all $\vp\in\core$, $\vp^+:=\PA\,\vp$,
\begin{equation*}
\SPb{\vp}{\Ynp{\gamma}\,\vp}
\grg\,e(0)\,\|\vp\|^2
-(c'/R)\,(e(0)+|e(\gamma_R)|)\,\|\vp^+\|^2
\,.
\end{equation*}
\end{lemma}

\begin{proof}
Let $\vp\in\core$ and set $\vp^+:=\PA\,\vp$. Since
$\nabla=\sum_{i=1,2}\mu_{i,R}\,\nabla\,\mu_{i,R}$ we have
$$
\SPn{\vp}{(\NPO{\gamma}+E\,\PA)\,\vp}
\,=\,
\sum_{i=1,2}
\SPb{\vp^+}{\PA\,\mu_{i,R}\,(\DA-\gamma/|\V{x}|+\HT)\,\mu_{i,R}\,\PA\,\vp^+}\,,
$$
where $\HT=\Hf+E$.
On account of Corollary~\ref{cor-compadre} with $\ve=1/R$
we thus have, for all $R\grg1$,
\begin{eqnarray}
\lefteqn{\nonumber
\SPn{\vp}{(\NPO{\gamma}+E\,\PA)\,\vp}
}
\\
&\grg&\nonumber
(1-1/R)\sum_{i=1,2}
\SPb{\vp^+}{\mu_{i,R}\,\PA\,(\DA+\HT)\,\PA\,\mu_{i,R}\,\vp^+}
\\
& &\;\,-(1+1/R)\,\SPb{\vp^+}{\mu_{1,R}\,\PA\,\tgV\,\PA\,\mu_{1,R}\,\vp^+}
\,-\,\big\|\,\mu_{2,R}^2\,\gamma/|\V{x}|\,\big\|_\infty\,\|\vp^+\|^2
\nonumber
\\
& &\qquad\qquad
\,-\,\sum_{i=1,2}\frac{3K(1+R)\,\|\nabla\mu_{i,1}\|^2}{R^2}
\,\big\|\,\HT^{1/2}\,\vp^+\,\big\|^2
\,.\label{laura1}
\end{eqnarray}
($K$ is the constant
appearing in Corollary~\ref{cor-compadre}.)
We set $C_\mu:=\sum_{i=1,2}\|\mu_{i,1}\|^2_\infty$ and apply
Corollary~\ref{cor-compadre} once more to obtain
\begin{eqnarray*}
\big\|\,\HT^{1/2}\,\vp^+\,\big\|^2
&=&
\sum_{i=1,2}\big\|\,\HT^{1/2}\,\mu_{i,R}\,\PA\,\vp^+\,\big\|^2
\\
&\klg&
2\sum_{i=1,2}\SPb{\vp^+}{\mu_{i,R}\,\PA\,\HT\,\PA\,\mu_{i,R}\,\vp^+}
\,+\,
\frac{2K\,C_\mu}{R}\,\big\|\,\HT^{1/2}\,\vp^+\,\big\|^2\,.
\end{eqnarray*}
Here we also estimated $1+1/R\klg2$.
This implies, for $R\grg 4K\,C_\mu$,
$$
\big\|\,\HT^{1/2}\,\vp^+\,\big\|^2\,\klg\,
4\sum_{i=1,2}\SPb{\vp^+}{\mu_{i,R}\,\PA\,\HT\,\PA\,\mu_{i,R}\,\vp^+}\,.
$$
Combining the previous estimate with \eqref{laura1} 
and setting $c:=1+24K\,C_\mu$
we arrive at
\begin{eqnarray*}
\lefteqn{
\SPn{\vp}{(\NPO{\gamma}+E\,\PA)\,\vp}
}
\\
&\grg&
(1-c/R)\,\SPb{\mu_{1,R}\,\vp^+}{(\NPO{\gamma_R}+E\,\PA)\,\mu_{1,R}\,\vp^+}
\\
&=&
\;+\,(1-c/R)\,\SPb{\mu_{2,R}\,\vp^+}{(\NPO{0}+E\,\PA)\,\mu_{2,R}}\,-\,
(\gamma/R)\,\|\vp^+\|^2
\\
&\grg&
(1-c/R)\,\big\{e(\gamma_R)\,\big\|\,\PA\,\mu_{1,R}\,\vp^+\,\|^2
+e(0)\,\big\|\,\PA\,\mu_{2,R}\,\vp^+\,\big\|^2\big\}
-(\gamma/R)\,\|\vp^+\|^2
\\
&\grg&
e(\gamma_R)\,\|\mu_{1,R}\,\vp^+\|^2\,+\,e(0)\,\|\mu_{2,R}\,\vp^+\|^2
\,-\,(c'/R)(e(0)+|e(\gamma_R)|)\,\|\vp^+\|^2
\,,
\end{eqnarray*}
where $\gamma_R$ is given by \eqref{def-gammaR} and 
$c'\in(0,\infty)$ is some universal constant. 
In the last step we used 
$\|\,[\PA,\mu_{2,R}]\,\|
\klg\|\nabla\mu_{2,1}\|_\infty/(2R)$.
We also assumed that $R>(\gamma+c\,\gcnp)/(\gcnp-\gamma)$,
which is equivalent to $\gamma_R<\gcnp$.
\end{proof}

\smallskip

\noindent
Next, we show that the conditions \eqref{eq-Y2a}
and \eqref{eq-Y2} required in
Lemma~\ref{le-MS-Y} are satisfied. 
We abbreviate
$$
\Delta(\gamma_R)\,:=\,e(0)\,-\,e(\gamma_R)\,=\,
\Thnp\,-\,\inf\spec(\NPO{\gamma_R})\,.
$$

\begin{lemma}\label{identity}
There is some constant $k_1\in(0,\infty)$ such that,
for all $\gamma\in(0,\gcnp)$,
$\V{G}$ fulfilling Hypothesis~\ref{hyp-G},
$a\in(0,1/2]$, all $F$ satisfying \eqref{hyp-F<>}, 
all sufficiently large $E>0$ (depending only on $d_{-1},d_0,d_1$, and
$\gamma$),
and all
$\vp\in\core$,
\begin{equation}
  \label{eq:22}
  \big|\,\Re\big[\SPb{\vp}{\big(\,e^F\,\Ynp{\gamma}\,e^{-F}\,-\,
\Ynp{\gamma}\,\big)\vp}\big]\,\big|\,\klg\,
k_1\,a^2\,\SPb{\vp}{\big(\Ynp{\gamma}+\Delta(\gamma_R)+\Thnp\big)\,\vp}\,,
\end{equation}
and 
\begin{equation}\label{iris}
\big|\,\SPb{\vp}{\big(\,e^F\,\Ynp{\gamma}\,e^{-F}\,-\,
\Ynp{\gamma}\,\big)\vp}\,\big|\,\klg\,
k_1\,a\,\SPb{\vp}{\big( 
c(\gamma)\,\Ynp{\gamma}+\Delta(\gamma_R)+\Thnp\big)\,\vp}\,,
\end{equation}
where $c(\gamma)=(\gcnp+\gamma)/(\gcnp-\gamma)$.
\end{lemma}

\begin{proof}
Let $\vp\in\core$
and let $\YnpF{\gamma}$ denote the operator obtained by replacing
the projections $\PA$ and $\PAm$ in $\Ynp{\gamma}$
by $\PAF$ and $\PAmF:=e^F\,\PAm\,e^{-F}$, respectively, i.e.
\begin{align*}
\YnpF{\gamma}&:=
\PAF\,\big(\DA-\tgV+\HT+\Delta(\gamma_R)\,\mu_{1,R}^2\big)\,\PAF
+\PAmF\,\HT\,\PAmF+
\Thnp\,\PAmF,
\end{align*}
where $\Delta(\gamma_R)=e(0)-e(\gamma_R)$.
Then $e^{-F}\,\DA\,e^{F}=\DA-i\valpha\cdot\nabla F$
implies 
\begin{align*}
\SPb{\vp}{e^F\,\Ynp{\gamma}\,e^{-F}\,\vp}\,&-\,
\SPb{\vp}{\YnpF{\gamma}\,\vp}
\,=\,
i\SPb{\vp}{\PAF\,\valpha\cdot\nabla F\,\PAF\,\vp}
\\
=\,&i\SPb{\vp}{\PA\,\valpha\cdot\nabla F\,\PA\,\vp}
\,+\,i\SPb{\vp}{\cK_F\,\valpha\cdot\nabla F\,\cK_F\,\vp}
\\
&-\,i\SPb{\vp}{\cK_F\,\valpha\cdot\nabla F\,\PA\,\vp}\,-\,
i\SPb{\vp}{\PA\,\valpha\cdot\nabla F\,\cK_F\,\vp}
\,.
\end{align*}
Since $\|\alpha\cdot\nabla F\|\klg a$ and $\|\cK_F\|\klg a\,J(a)/2$
we thus obtain
\begin{equation}\label{bea41}
\big|\SPb{\vp}{e^F\,\Ynp{\gamma}\,e^{-F}\,\vp}\,-\,
\SPb{\vp}{\YnpF{\gamma}\,\vp}\big|\,\klg\,
a\,(1+a\,J(a)+a^2\,J(a)^2/4)\,\|\vp\|^2\,.
\end{equation}
Since $\Re[i\SPn{\vp}{\PA\,\valpha\cdot\nabla F\,\PA\,\vp}]=0$
we further have
\begin{equation}\label{bea42}
\big|\,\Re\big[\,\SPb{\vp}{e^F\,\Ynp{\gamma}\,e^{-F}\,\vp}\,-\,
\SPb{\vp}{\YnpF{\gamma}\,\vp}\,\big]\,\big|\,\klg\,
a^2\,(J(a)+a\,J(a)^2/4)\,\|\vp\|^2\,.
\end{equation}
Assuming $E\grg(4d_1\,J(a))^2$ we next 
apply Corollary~\ref{cor-compadre}
(Estimate \eqref{eq-cO-a^2} and its obvious analogue
for $\PAm$)
to each of the six terms in 
$\Re[\YnpF{\gamma}-\Ynp{\gamma}]$
(involving the operators
$\DA,\tgV,\PApm\HT\PApm$, $\Delta(\gamma_R)\,\mu_{1,R}^2$, and $\Thnp$,
respectively). 
As a result we find some universal constant, $k_2\in(0,\infty)$,
such that, for all $\vp\in\core$,
\begin{eqnarray*}
\lefteqn{
\big|\,\Re\big[\,\SPn{\vp}{\YnpF{\gamma}\,\vp}\,-\,
\SPn{\vp}{\Ynp{\gamma}\,\vp}\,\big]\,\big|
}
\\
&\klg&
k_2\,(a\,J(a))^2\,\big(
\SPn{\vp}{\HT\,\vp}+\SPn{\vp}{\PA\,\HT\,\PA\,\vp}+
\SPn{\vp}{\PAm\,\HT\,\PAm\,\vp}
\big)
\\
& &\qquad\;+\,k_2\,(a\,J(a))^2\,(\Delta(\gamma_R)+\Thnp)\,\|\vp\|^2
\,.
\end{eqnarray*}
Here the off-diagonal terms in $\SPn{\vp}{\HT\,\vp}$
can be estimated as 
\begin{equation}\label{Hf-off-diag}
2\,\Re\SPn{\vp}{\PA\,\HT\,\PAm\,\vp}
\,\klg\,\|\HT^{1/2}\,\PA\,\vp\|^2\,+\,\|\HT^{1/2}\,\PAm\,\vp\|^2\,.
\end{equation}
Therefore, we arrive at
\begin{eqnarray}
\lefteqn{\nonumber
\big|\,\Re\big[\,\SPn{\vp}{\YnpF{\gamma}\,\vp}\,-\,
\SPn{\vp}{\Ynp{\gamma}\,\vp}\,\big]\,\big|
}
\\
&\klg&\label{bea11}
6k_2\,(a\,J(a))^2\,\Big(\,\frac{1}{2}\,\SPn{\vp}{\PA\,\HT\,\PA\,\vp}\,+\,
\SPn{\vp}{\PAm\,\HT\,\PAm\,\vp}
\Big)
\\
& &\qquad\;+\,k_2\,(a\,J(a))^2\,(\Delta(\gamma_R)+\Thnp)\,\|\vp\|^2
\,.\nonumber
\end{eqnarray}
Now, we assume that $E$ is so large that
$\wt{H}:=\PA(\DA-\tgV+(1/2)\,\Hf+E)\PA\grg0$. In fact,
this is possible according to Theorem~\ref{thm-sb-np}
with $\delta=1/2$.
Then we can add $6k_2\,(a\,J(a))^2\,\SPn{\vp}{\{
\wt{H}+\Delta(\gamma_R)\,\PA\,\mu_{1,R}^2\,\PA+\Thnp\,\PAm\}\,\vp}$
to the right hand side of \eqref{bea11}
and combine the resulting estimate with \eqref{bea42}
to obtain \eqref{eq:22}.

In order to derive \eqref{iris} 
we apply Corollary~\ref{compadre} (Estimate~\eqref{eq-cO}
with $\ve=a$ and its obvious analogue for $\PAm$)
to each of
the six terms in $\YnpF{\gamma}-\Ynp{\gamma}$.
Proceeding in this way we find some universal constant,
$k_3\in(0,\infty)$, such that, for every $\vp\in\core$,
\begin{align*}
\big|\SPb{\vp}{(\YnpF{\gamma}-\Ynp{\gamma})\,\vp}\big|
\,&\klg\,
k_3\,a\,\SPb{\vp}{(\Ynp{0}+\HT+\PA\,\tgV\,\PA)\,\vp}
\\
&\qquad\qquad
\,+\,
k_3\,a\,(\Delta(\gamma_R)+\Thnp)\,\|\vp\|^2\,.
\end{align*}
As above we argue that $\Ynp{0}+\HT\klg\,k_4\,\Ynp{0}=
k_4(\Ynp{\gamma}+\PA\,\tgV\,\PA)$
and it follows from Theorem~\ref{thm-sb-np}
that $\PA\,\tgV\,\PA\klg2\gamma(\gcnp-\gamma)^{-1}\,\Ynp{\gamma}$,
provided $E>0$ is sufficiently large depending
on $d_{-1},d_0,d_1$, and $\gamma$.
Combining these remarks with \eqref{bea41} we arrive at \eqref{iris}.
\end{proof}

\smallskip

\noindent
In the following lemma we verify another assumption
made in Lemma~\ref{le-MS-Y}.

\begin{lemma}\label{conservation-formdomain}
There exist constants $c_1,c_2\in(0,\infty)$ 
such that, for all $F:\RR^3\to\RR$ satisfying \eqref{hyp-F<>} 
and all $\vp \in \core$, 
\begin{equation}
  \label{mona7}
\SPb{e^{F}\,\vp}{\Ynp{0}\,e^{F}\,\vp}\,\klg\,
c_1\,\|e^{F}\|^2\,
\SPn{\vp}{\Ynp{0}\,\vp}+c_2\,\|e^{F}\|^2\,\|\vp\|^2\,.
\end{equation}
In particular, $e^{F}\form(\Ynp{\gamma})\subset\form(\Ynp{\gamma})$,
for every $\gamma\in[0,\gcnp)$.
\end{lemma}

\begin{proof}
It is clear that we only have to comment
on the unbounded terms in $\Ynp{0}$.
In \cite[Equation~(4.24) and the succeding paragraphs]{MatteStockmeyer2008a}
we proved that
\begin{equation}\label{heidi99}
\SPb{\vp}{e^F\,\PApm\,(\pm\DA)\,e^F\,\vp}\,\klg\,c_3\,
\|e^F\|^2\,\SPn{\vp}{\PApm\,(\pm\DA)\,\vp}\,+\,c_4\,\|e^F\|^2\,
\|\vp\|^2\,,
\end{equation}
for every $\vp\in\core$.
We derived this bound in \cite{MatteStockmeyer2008a}
for classical vector potentials.
The proof works, however, also for the quantized vector potential
without any change. Moreover, we only treated the choice
of the plus sign in \eqref{heidi99}. But again an obvious
modification of the 
proof in \cite{MatteStockmeyer2008a} shows that \eqref{heidi99}
is still valid when we choose the minus sign.
(This will actually be necessary only in the next subsection where
we treat the semi-relativistic Pauli-Fierz operator.)
On account of \eqref{Hf-off-diag} 
it thus remains to show that 
$\|\HT^{1/2}\,\PApm\,e^F\,\vp\,\|\klg c_5\,\|e^F\|\,\|\HT^{1/2}\,\vp\|$.
This follows, however, immediately from \eqref{eva99}
which implies
$
\|\HT^{1/2}\,\PApm\,e^F\,\vp\,\|\klg
(1+\|S_{1/2}\|/2)\,\|e^F\,\HT^{1/2}\,\vp\|
$. From these remarks we readily derive the asserted estimate
which shows that $e^{F}\form(\Ynp{0})\subset\form(\Ynp{0})$
holds true. But from Theorem~\ref{thm-sb-np} we know that
$\form(\Ynp{\gamma})=\form(\Ynp{0})$, for every $\gamma\in[0,\gcnp)$.
\end{proof}

\bigskip

\begin{proof}[Proof of Theorem~\ref{thm-el-np}]
Assume that $\gamma<\gcnp$ and
let $I\subset\RR$ be a compact interval with
$\max I<\Thnp$. We fix some $E\in[1,\infty)$ and set 
$I_E:=I+E$. In the following we assume that $E$ is so large 
that Lemmata~\ref{boundbelow} and~\ref{identity} are applicable.
(Then $E$ depends on $d_{-1},d_0,d_1$, and $\gamma$.)
Let $k_1$ be the constant appearing in the statement of 
Lemma~\ref{identity}, $\lambda:=\max I_E$, and $e(0)=\Thnp+E$.
We assume that $a\in(0,1/2]$
is so small that $k_1\,a\,(\gcnp+\gamma)(\gcnp-\gamma)^{-1}<1/2$ and
$\ve:=\{1-(\lambda/e(0))-5k_1\,a^2\}/4>0$.
On account of Lemma~\ref{boundbelow} we may fix some $R\grg1$
such that 
$b:=\inf\spec(\Ynp{\gamma})\grg e(0)-\ve$, which implies
$1/b<(1/e(0))(1+2\ve)$. 
(We can choose $R=c_1\,e(0)/\ve$, for some universal constant $c_1$.)
By virtue of Lemmata~\ref{identity} and~\ref{conservation-formdomain} 
we can then
apply Lemma~\ref{le-MS-Y} with $g(a):=k_1\,a^2$,
$h(a)=k_1\,a^2\,(\Delta(\gamma_R)+\Thnp)$, and $J:=I_E+[-b\,\ve,b\,\ve]$.
In view of Theorem~\ref{thm-sb-np} we can further assume that
$\Delta(\gamma_R)\klg\Thnp+E=e(0)$. From these remarks we
infer that $b-\max J-b\,g(a)-h(a)\grg b\,\ve$. 
This ensures that the second condition in \eqref{hyp-XY}
is fulfilled with $C'\klg1/(b\,\ve)$.
The first bound in \eqref{hyp-XY} is also valid
since 
$\Xnp{\gamma}-\Ynp{\gamma}=\Delta(\gamma_R)\,\PA\,\mu_{1,R}^2\,\PA$
and
$\|e^F\,\PA\,\mu_{1,R}^2\,\PA\|\klg\|\PAF\|\,\|e^F\,\mu_{1,R}^2\|
\klg\const\cdot e^{2aR}$.
Then Lemma~\ref{le-ed-BFS} with $\dist(I,J^c)=b\ve$
and $|J|=|I|+2b\ve$ implies that
$$
\big\|\,e^{a|\V{x}|}\,\id_{I_E}(\Xnp{\gamma})\,\PA\,\big\|\,\klg\,
\const\cdot\big\{\,1+(1+|I|)/(b\ve)+1/(b\ve)^2\,\big\}\,
e(0)\,e^{2aR}\,\,.
$$
Since $1/b\klg2/e(0)\klg2$ and
$\PA\,\Xnp{\gamma}=\Xnp{\gamma}\,\PA=(\NPO{\gamma}+E\,\PA)\oplus0$
this proves
Theorem~\ref{thm-el-np}.
(Keeping track of all conditions imposed on $E=E(\gamma,d_{-1},d_0,d_1)$
we see that we can choose $-E$ proportional to the term in the
second line in \eqref{eq-sb-np}, whence $E\to0$, $d_i\to0$, $i\in\{-1,0,1\}$.)
\end{proof}


\subsection{The semi-relativistic Pauli-Fierz operator: Localization}
\label{ssec-ed-PF}

\noindent
Again we employ the partition of unity $\mu_{1,R}^2+\mu_{2,R}^2=1$
constructed in the first paragraph of Subsection~\ref{ssec-ed-np}.
We set
$$
\epsilon(\gamma)\,\equiv\,\epsilon(\gamma,\V{G})
\,:=\,\inf\spec(\PF{\gamma})\,,\qquad \gamma\in[0,\gcPF]\,,
$$
so that $\epsilon(0)=\ThPF$,
and apply Lemma~\ref{le-ed-BFS} with
\begin{eqnarray*}
\XPF{\gamma}&=&\PF{\gamma}\,,
\\
\YPF{\gamma}&=&\PF{\gamma}\,
+\,(\epsilon(0)-\epsilon(\gamma_R))\,\mu_{1,R}^2\,,
\end{eqnarray*}
where
$$
\gamma_R\,:=\,\gamma/(1-1/R)\,,\qquad R\grg1\,.
$$
Of course, $\XPF{\gamma}$ and $\YPF{\gamma}$ are self-adjoint
on the same domain and both admit $\core$ as a form core.
The remaining conditions of Lemma~\ref{le-ed-BFS} are easier
to verify than in the previous subsection since only the kinetic
energy term in the semi-relativistic Pauli-Fierz operator is
non-local.

\begin{lemma}\label{le-lb-XPF}
There is some $C\in(0,\infty)$ such that, for all
$\gamma\in(0,\gcPF)$, $R\grg\gcPF/(\gcPF-\gamma)$, 
$\V{G}$ fulfilling Hypothesis~\ref{hyp-G}, and $\vp\in\core$,
\begin{equation}
\SPn{\vp}{\PF{\gamma}\,\vp}\,\grg\,
\epsilon(\gamma_R)\,\|\mu_{1,R}\,\vp\|^2\,+\,\epsilon(0)
\,\|\mu_{2,R}\,\vp\|^2\,-\,\frac{\epsilon(0)+\epsilon(\gamma_R)+C}{R}
\:\|\vp\|^2\,.
\end{equation}
\end{lemma}

\begin{proof}
Let $\vp\in\core$.
We write $|\DA|=\PA\,\DA\,\PA-\PAm\,\DA\,\PAm$ and obtain
by means of Corollary~\ref{cor-compadre} 
(and its obvious analogue for $\PAm$)
\begin{eqnarray*}
\SPn{\vp}{|\DA|\,\vp}&=&
\sum_{\sharp=\pm}\sum_{i=1,2}\SPb{\vp}{\PAsharp\,\mu_i\,(\sharp1)\,
\DA\,\mu_i\,\PAsharp\,\vp}
\\
&\grg&
(1-1/R)\sum_{i=1,2}\SPb{\vp}{\mu_i\,|\DA|\,\mu_i\,\vp}\,-\,
\frac{2C_\mu K(1+R)}{R^2}\,\|\vp\|^2\,,
\end{eqnarray*}
where $C_\mu=\|\nabla\mu_{1,1}\|_\infty+\|\nabla\mu_{2,1}\|_\infty$.
The remaining term, $-\gamma/|\V{x}|+\HT$, in $\PF{\gamma}$
commutes with $\mu_1$ and $\mu_2$, so the assertion becomes
evident.
\end{proof}

\smallskip

\noindent
In the next lemma we verify the conditions \eqref{eq-Y2a}
and \eqref{eq-Y2}
of Lemma~\ref{le-MS-Y}. In contrast to the previous
subsection we can now choose $g=0$ in \eqref{eq-Y2}.
This will result in an estimate on the exponential decay
rate for the semi-relativistic Pauli-Fierz operator
that does not depend on the values of $d_{-1},\ldots,d_2$.

\begin{lemma}\label{le-O(a)-PF}
For all $a\in(0,1)$, $F$ satisfying \eqref{hyp-F<>},
$\gamma\in(0,\gcPF)$,
$\V{G}$ fulfilling Hypothesis~\ref{hyp-G}, and $\vp\in\core$,
\begin{equation}\label{eq-O(a)-PF}
\big|\,\Re\big[\SPb{\vp}{(e^F\,\PF{\gamma}\,e^{-F}-\PF{\gamma})\,\vp}\big]
\,\big|\,
\klg\,(3/2)\,a^2\,J(a)^2\,\|\vp\|^2\,.
\end{equation}
Moreover, for every $\ve>0$, there is some constant, 
$C(a,\gamma,\ve)\in(0,\infty)$,
such that
\begin{equation}\label{iris-PF}
\big|\SPb{\vp}{(e^F\,\PF{\gamma}\,e^{-F}-\PF{\gamma})\,\vp}\big|
\,\klg\,
\ve\,\SPb{\vp}{\PF{\gamma}\,\vp}\,+\,C(a,\gamma,\ve)\,\|\vp\|^2\,.
\end{equation}
\end{lemma}

\begin{proof}
On $\core$
the operator $\Re\big[e^F\,\PF{\gamma}\,e^{-F}-\PF{\gamma}\big]$ 
appearing on the left side of 
\eqref{eq-O(a)-PF} equals
\begin{eqnarray*}
\lefteqn{
\Re\big[\,e^F\,|\DA|\,e^{-F}-|\DA|\,\big]
\,=\,
\frac{1}{2}\,\big[\,e^{-F}\,,\,[\,|\DA|\,,\,e^F]\,\big]
}
\\
&=&
\frac{1}{2}\,\big[\,e^{-F}\,,\,\DA\,[\sgn(\DA)\,,\,e^F]
\,-\,i\valpha\cdot(\nabla F)\,e^F
\,\sgn(\DA)\,\big]
\\
&=&
\frac{1}{2}\,\DA\,\big[\,e^{-F}\,,\,[\sgn(\DA)\,,\,e^F]\,\big]
\,-\,i\valpha\cdot\nabla F\,(\cK_{0,-F}+\cK_{0,F})
\,,
\end{eqnarray*}
where we use the notation \eqref{def-cK}.
On account of Lemmata~\ref{compadre} and~\ref{le-dc}
this implies \eqref{eq-O(a)-PF}.
Moreover, since
$$
e^F\,|\DA|\,e^{-F}-|\DA|\,=\,-2\,\DA\,\cK_F+i\valpha\cdot(\nabla F)\,
e^F\,\sgn(\DA)\,e^{-F}
$$
holds true on $\core$, the left hand side of \eqref{iris-PF}
is less than or equal to
\begin{align*}
\ve_1\,\SPb{\vp&}{|\DA|\,\vp}\,+\,\ve_1^{-1}\,
\big\|\,|\DA|^{1/2}\,\cK_F\,\big\|^2\,\|\vp\|^2+
a\,\big\|\,e^F\,\sgn(\DA)\,e^{-F}\,\big\|^2\,\|\vp\|^2
\\
&\klg\,
\ve_1\,\const(\gamma)\,\SPb{\vp}{\PF{\gamma}\,\vp}
\,+\,\const(a,\ve_1)\,\|\vp\|^2\,,
\end{align*}
for every $\ve_1>0$. This proves \eqref{iris-PF}.
\end{proof}

\begin{lemma}\label{le-fdom-PF}
There exist constants, $c_1,c_2\in(0,\infty)$, such that,
for all $a\in(0,1)$ and $F$ satisfying \eqref{hyp-F<>},
$$
\SPb{e^{F}\,\vp}{\YPF{0}\,e^F\,\vp}\,\klg\,c_1\,\|e^F\|^2\,
\SPb{\vp}{\YPF{0}\,\vp}
\,+\,
c_2\,\|e^F\|^2\,\|\vp\|^2\,,\qquad \vp\in\core\,.
$$
In particular, $e^F\,\form(\YPF{\gamma})\subset\form(\YPF{\gamma})$,
for every $\gamma\in(0,\gcPF)$.
\end{lemma}

\begin{proof}
Of course,
$\|\mu_{1,R}\,e^F\,\vp\|^2\klg\|e^F\|^2\,\|\mu_{1,R}\,\vp\|^2$ and,
since $\Hf$ and $e^F$ commute, 
$\|\Hf^{1/2}\,e^F\,\vp\|^2\klg\|e^F\|^2\,\|\Hf^{1/2}\,\vp\|^2$.
To conclude we write 
$|\DA|=\PA\,\DA-\PAm\,\DA$ and again employ the bound
\eqref{heidi99} derived in \cite{MatteStockmeyer2008a}.
\end{proof}

\smallskip

\begin{proof}[Proof of Theorem~\ref{thm-el-PF}]
Let $\gamma\in(0,\gcPF)$, let $I\subset(-\infty,\ThPF)$
be some compact interval, and let 
$a\in(0,1)$ satisfy 
$\ve:=(\ThPF-\max I-(3/2)\,a^2\,J(a)^2)/3>0$.
By virtue of Lemma~\ref{le-lb-XPF} we may choose $R\grg \gcPF/(\gcPF-\gamma)$
so large that $\YPF{\gamma}\grg \epsilon(0)-\ve$.
On account of Lemmata~\ref{le-O(a)-PF} and~\ref{le-fdom-PF}
we may apply Lemma~\ref{le-MS-Y} with $J=I+[-\ve,\ve]$
and $h(a)=(3/2)\,a^2\,J(a)^2$, $g(a)=0$, and $c(a)=1/4$.
It ensures that the second condition in \eqref{hyp-XY}
is fulfilled with $C'=1/\ve$.
Moreover, 
$\|e^F(\XPF{\gamma}-\YPF{\gamma})\|=(\epsilon(0)-\epsilon(\gamma_R))
\,\|e^F\,\mu_{1,R}^2\|$ and $\|e^F\,\mu_{1,R}^2\|\klg e^{2aR}$,
so the first condition in \eqref{hyp-XY}
is fulfilled also, with $C= (\epsilon(0)-\epsilon(\gamma_R)\,)e^{2aR}$.
Therefore, Theorem~\ref{thm-el-PF}
is a consequence of Lemma~\ref{le-ed-BFS}
and \eqref{gustav2}, which implies that 
$|\epsilon(\gamma_R)|\klg\ThPF+\const\cdot d_1^2$.
\end{proof}


\appendix

\section{The ground  state energy and ionization
threshold for small field strength}\label{O(g)}

\noindent
In this appendix we prove the perturbative estimates
on the ground state energies and ionization thresholds
of the no-pair and semi-relativistic Pauli-Fierz operators
asserted in Remarks~\ref{rem-O(g)-np} and~\ref{rem-O(g)-PF},
respectively.
In the whole appendix we always assume that 
$$
0\,<\,d_1\,\klg\,1\,,\qquad 0\,<\,d_*^2\,:=\,d_0^2\,+\,2d_{-1}^2\,\klg\,1\,.
$$
Moreover, we fix some value of $E$ such that
\begin{equation}\label{judith}
2d_1/E^{1/2}\,\klg\,1/2\,.
\end{equation}
We start with the semi-relativistic Pauli-Fierz operator.

\begin{proof}[Proof of Remark~\ref{rem-O(g)-PF}]
For every $\gamma\in(0,\gcPF)$,
we let $E_\el^\mathrm{C}(\gamma)$
denote the (positive) ground state energy
of Chandrasekhar's operator,
$|\DO|-\gamma/|\V{x}|$,
and $\phi_\el^{\mathrm{C}}(\gamma)$ a corresponding
ground state
eigenfunction. For $\ve>0$, we set
$\gamma_\ve:=\frac{1}{1+\ve}\,\gamma$.
Using the minimax principle
and Kato's inequality,
which can be written as 
$1/|\V{x}|\klg(\gcPF-\gamma)^{-1}\big(|\DO|-\gamma/|\V{x}|\big)$,
it is easy to see that
$$
0\,\klg\, 
E_\el^\mathrm{C}(\gamma_\ve)-E_\el^\mathrm{C}(\gamma)\, \klg\, 
\frac{\ve}{1+\ve}\,c(\gamma)\,E_\el^\mathrm{C}(\gamma)\,,\qquad
c(\gamma)\,:=\,\frac{\gamma}{\gcPF-\gamma}
\,.
$$
Next, let $\vp\in H^{1/2}(\RR^3_\V{x})\otimes\sC_0$.
On account of Lemma~\ref{le-DA-D0}
we have
\begin{eqnarray}
\lefteqn{\nonumber
\big|\,
\SPb{\vp}{\PApm\,(\pm\DO)\,\PApm\,
\vp}\,-\,
\SPb{\vp}{\POpm\,
(\pm\DO)\,\POpm\,\vp}
\,\big|
}
\\
&\klg&\label{knut1}
\ve\,\SPb{\vp}{\POpm\,
(\pm\DO)\,\POpm\,\vp}
+\frac{(1+\tfrac{1}{\ve})
\,\bigO(d_*^2)}{(1-2d_1/E^{1/2})^2}\,
\big\|\,\HT^{1/2}\,\vp\,\big\|^2
\,.
\end{eqnarray}
Moreover, by virtue of Lemma~\ref{le-sgn} we find, for
every $\delta>0$,
\begin{eqnarray}
\lefteqn{
\big|\SPb{\vp}{\PApm\,\valpha\cdot\V{A}\,\PApm\,
\vp}\big|\nonumber
}
\\
&\klg&\nonumber
\|\vp\|\,\big\|\,\valpha\cdot\V{A}\,\HT^{-1/2}\,\PApm\,\big\|
\,\Big(1
\,+\,\frac{1}{2}\cdot\frac{d_1/E^{1/2}}{1-2d_1/E^{1/2}}\Big)
\big\|\,\HT^{1/2}\,\vp\,\big\|
\\
&\klg&\label{knut2}
C\,d_*\,
\big(\delta\,\SPn{\vp}{\HT\,\vp}+\delta^{-1}\,\|\vp\|^2\big)\,,
\end{eqnarray}
where $C\in(0,\infty)$ is some universal constant.
Here we used \eqref{rb-a}, \eqref{rb-ad}, $d_*^2:=d_0^2+2d_{-1}^2$,
and \eqref{judith}.
Since $|\DA|=\PA\,\DA\,\PA-\PAm\,\DA\,\PAm$
and $\HT\,\Omega=E\,\Omega$
the above estimates with 
$\vp=\phi_\el^{\mathrm{C}}(\gamma_\ve)\otimes\Omega$ and $\delta=1$
show that
\begin{eqnarray}
\lefteqn{\nonumber
\SPb{\phi_\el^{\mathrm{C}}(\gamma_\ve)\otimes\Omega}{
(|\DA|-\tgV+\HT)\,\phi_\el^{\mathrm{C}}(\gamma_\ve)\otimes\Omega}
}
\\
&\klg&\nonumber
(1+\ve)\,\SPb{\phi_\el^{\mathrm{C}}(\gamma_\ve)}{
\big(|\DO|-{\gamma}_\ve/|\V{x}|\big)
\,\phi_\el^{\mathrm{C}}(\gamma_\ve)}\,+\,E
\\
& &\qquad\qquad\qquad\quad
+\,(1+\tfrac{1}{\ve})\,\bigO(d_*^2)\,(1+E)\,+\,\bigO(d_*)\,(1+E)\label{knut}
\\
&\klg&
(1+\ve)(1+c(\gamma))\,E_\el^\mathrm{C}(\gamma)\,+\,E
\,+\,\big(\tfrac{1}{\ve}\,\bigO(d_*^2)+\bigO(d_*)\big)\,(1+E)
\,,\nonumber
\end{eqnarray}
provided \eqref{judith} holds true.
If we choose $\ve=d_*$,
then we find some $C_\gamma'\in(0,\infty)$ such that
$$
\inf\spec(\PF{\gamma})\,\klg\,E_\el^\mathrm{C}(\gamma)\,+\,
C_\gamma'\,(d_1+d_*)\,.
$$
If we set $\gamma=0$ in \eqref{knut} and replace
$\phi_\el^{\mathrm{C}}(\gamma_\ve)$ by some normalized
$\chi\in C_0^\infty(\RR^3_\V{x},\CC^4)$ satisfying
$\SPn{\chi}{|\DO|\,\chi}\klg 1+\ve$, then we obtain
$$
\SPb{\chi\otimes\Omega}{(\PF{0}+E)\chi\otimes\Omega}\,\klg\,
(1+\ve)^2\,+\,E
\,+\,\big(\tfrac{1}{\ve}\,\bigO(d_*^2)+\bigO(d_*)\big)\,(1+E)\,.
$$
Choosing $\ve=d_*$ as above we 
we find some $C_0'\in(0,\infty)$ such that
$$
1\,\klg\,\ThPF\,\klg\,1\,+\,C_0'\,(d_1+d_*)\,.
$$
It remains to derive the lower bound on $\inf\spec(\PF{\gamma})$.
To this end we set $\wt{\gamma}_\ve:=\gamma/(1-\ve)$, for some
$\ve>0$ such that $\wt{\gamma}_\ve<\gcPF$. Moreover, we
choose $\delta=1/(2Cd_*)$ in \eqref{knut2}.
Then \eqref{knut1} and \eqref{knut2}
permit to get, for every $\vp\in\core$,
\begin{align*}
\SPb{\vp}{(\PF{\gamma}+&E)\,\vp}
\grg\,
(1-\ve)\,\SPb{\vp}{\big(|\DO|-\wt{\gamma}_\ve/|\V{x}|\big)\,\vp}
\\
&-\,\big(1-(1+1/\ve)\,\bigO(d_*^2)-1/2\big)\,\SPn{\vp}{(\Hf+E)\,\vp}
\,-\,2\,C^2\,d_*^2\,\|\vp\|^2\,.
\end{align*}
Here we again made use of \eqref{judith}.
So, choosing $\ve=d_*$
and using
$$
\SPb{\vp}{\big(|\DO|-\wt{\gamma}_\ve/|\V{x}|\big)\,\vp}\,\grg\,
E_\el^{\mathrm{C}}(\wt{\gamma}_\ve)\,\grg\,
\Big(1\,-\,\frac{\ve}{1-\ve}\cdot\frac{\gamma}{\gcPF-\gamma}\Big)\,
E_\el^{\mathrm{C}}(\gamma)\,,
$$
which is a straightforward consequence of the minimax principle,
we find some $C_\gamma''\in(0,\infty)$ such that
$$
\SPb{\vp}{\PF{\gamma}\,\vp}\,\grg\,
(1-C_\gamma''\,d_*)\,E_\el^{\mathrm{C}}(\gamma)
\,-\,\bigO(d_*^2)\,\|\vp\|^2\,,
$$
for all sufficiently small values of $d_*$.
\end{proof}

\smallskip

\begin{proof}[Proof of Remark~\ref{rem-O(g)-np}]
For $\gamma\in(0,\gcnp)$, we let
$E_\el^\mathrm{B}(\gamma)$
and $\phi_\el^\mathrm{B}(\gamma)$ 
denote the ground state energy and a 
normalized ground state
eigenfunction of the Brown-Ravenhall operator,
that is,
$$
\PO\,(\D{0}-\tfrac{\gamma}{|\V{x}|})\,\PO\,\phi_\el^\mathrm{B}(\gamma)\,
=\,E_\el^\mathrm{B}(\gamma)\,\phi_\el^\mathrm{B}(\gamma)\,.
$$
It is known that 
$E_\el^\mathrm{B}(\gamma)\in[1-\gamma_\ve,1)$ \cite{Tix1998}.
We set 
$$
\gamma_\ve\,:=\,(1-\ve)\gamma/(1+\ve)\,,\qquad\ve\in(0,1]\,.
$$
Then a standard argument based on the inequality \cite{EPS1996}
$$
\gcnp\,\SPn{\phi}{|\V{x}|^{-1}\,\phi}\klg\SPn{\PO\,\phi}{\DO\,\PO\,\phi}\,,
\qquad \phi\in H^{1/2}(\RR^3,\CC^4)\,,
$$ 
and the minimax principle
shows that 
\begin{equation}\label{minmax-B}
0\,\klg\, E_\el^\mathrm{B}(\gamma_\ve)-E_\el^\mathrm{B}(\gamma)
\,\klg\, 
\frac{\ve}{1+\ve}\cdot\frac{\gamma}{\gcnp-\gamma}\,E_\el^\mathrm{B}(\gamma)\,.
\end{equation}
Using \eqref{judith}, \eqref{knut1}, \eqref{knut2} with $\delta=1$, and
$$
\PO\,\phi_\el^\mathrm{B}(\gamma_\ve)=\phi_\el^\mathrm{B}(\gamma_\ve)\,,
\qquad
\|\phi_\el^\mathrm{B}(\gamma_\ve)\otimes\Omega\|=1\,,
\qquad \HT\,\phi_\el^\mathrm{B}(\gamma_\ve)\otimes\Omega
=E\,\phi_\el^\mathrm{B}(\gamma_\ve)\otimes\Omega\,,
$$
we deduce that
\begin{eqnarray}
\lefteqn{\nonumber
\SPb{\phi_\el^\mathrm{B}(\gamma_\ve)\otimes\Omega}{
\PA\,\DA\,\PA\,\phi_\el^\mathrm{B}(\gamma_\ve)\otimes\Omega}
}
\\
&\klg&\nonumber
(1+\ve)\,\SPb{\phi_\el^\mathrm{B}(\gamma_\ve)\otimes\Omega}{
\DO\,\phi_\el^\mathrm{B}(\gamma_\ve)\otimes\Omega}
\\
& &\;
+\,
(1+1/\ve)
\,\bigO(d_*^2)\,E\,
+\,
C\,d_*\,\label{ina1}
(E+1)
\,.
\end{eqnarray}
Moreover, since $\PA=(1/2)\,\id+(1/2)\,\sgn(\DA)$,
Lemma~\ref{le-sgn} yields
\begin{eqnarray*}
\lefteqn{
\SPb{\phi_\el^\mathrm{B}(\gamma_\ve)\otimes\Omega}{
\PA\,\HT\,\PA\,\phi_\el^\mathrm{B}(\gamma_\ve)\otimes\Omega}
}
\\
&\klg&
(1+\ve')\,\big\|\,\PA\,\HT^{1/2}\,
\phi_\el^\mathrm{B}(\gamma_\ve)\otimes\Omega\,\big\|^2
\\
& &\qquad\qquad
+\,(1+\tfrac{1}{\ve'})\,\big\|\,S_{1/2}
\,\HT^{1/2}\,\phi_\el^\mathrm{B}(\gamma_\ve)\otimes\Omega\,\big\|^2
\\
&\klg&E\,\big\|\,\PA\,\phi_\el^\mathrm{B}(\gamma_\ve)\otimes\Omega\,\big\|^2
\,+\,
\big(\ve'\,E+(1+1/\ve')\,\bigO(d_1^2)\big)
\,\|\phi_\el^\mathrm{B}(\gamma_\ve)\otimes\Omega\|^2\,.
\end{eqnarray*}
Using Lemma~\ref{le-DA-D0}, Kato's inequality,
$\||\V{x}|^{-1/2}\,\vp\|^2\klg(\pi/2)\|\,|\DO|^{-1/2}\,\vp\|^2$,
and \eqref{judith}, we find
for the potential energy 
\begin{eqnarray*}
\lefteqn{
-\SPb{\phi_\el^\mathrm{B}(\gamma_\ve)\otimes\Omega}{\gamma/|\V{x}|\,
\phi_\el^\mathrm{B}(\gamma_\ve)\otimes\Omega}
}
\\
&\klg&
-(1-\ve)\,\gamma\,\big\|\,|\V{x}|^{-1/2}\,\phi_\el^\mathrm{B}(\gamma_\ve)
\otimes\Omega\,\big\|^2
\\
& &
\qquad\qquad-\,
(1-\tfrac{1}{\ve})\,\gamma\,\big\|
\,|\V{x}|^{-1/2}\,(\PA-\PO)\,\phi_\el^\mathrm{B}(\gamma_\ve)
\otimes\Omega\,\big\|^2
\\
&\klg&
-(1+\ve)\,\SPb{\phi_\el^\mathrm{B}(\gamma_\ve)\otimes\Omega}{
\gamma_\ve/|\V{x}|\,\phi_\el^\mathrm{B}(\gamma_\ve)\otimes\Omega}
+
\tfrac{1}{\ve}\,\bigO(d_*^2)\,E\,
\|\,\phi_\el^\mathrm{B}(\gamma_\ve)\otimes\Omega\,\|^2.
\end{eqnarray*}
Putting the estimates above together 
we arrive at
\begin{eqnarray}
\lefteqn{\nonumber
\SPb{\phi_\el^\mathrm{B}(\gamma_\ve)\otimes\Omega}{\PA\,(
\DA-\gamma/|\V{x}|+\Hf+E)\,\PA\phi_\el^\mathrm{B}(\gamma_\ve)\otimes\Omega}
}
\\
&\klg&
(1+\ve)\,\nonumber
\SPb{\phi_\el^\mathrm{B}(\gamma_\ve)\otimes\Omega}{\PO\,(\DO-
\gamma_\ve/|\V{x}|)\,\PO\,\phi_\el^\mathrm{B}(\gamma_\ve)\otimes\Omega}
\\
& &\nonumber
\;+\,E\,\big\|\,\PA\,\phi_\el^\mathrm{B}(\gamma_\ve)\otimes\Omega\,\big\|^2
\\
& &\nonumber
\;+\;\big\{\ve'\,E+(1+1/\ve')\,\bigO(d_1^2)+(1+1/\ve)\,E\,\bigO(d_*^2)
\\
& &\qquad\qquad \qquad\qquad\qquad \qquad \label{tina1}
+\,C\,d_*\,(1+E)\,\big\}\,
\|\,\phi_\el^\mathrm{B}(\gamma_\ve)\otimes\Omega\,\|^2
\,.
\end{eqnarray}
On the other hand,
\begin{align}
\|\phi_\el^\mathrm{B}(\gamma_\ve)&\otimes\Omega\|^2\nonumber
\,=\,\big\|\,\PO\,\phi_\el^\mathrm{B}(\gamma_\ve)\otimes\Omega\,\big\|^2
\\\nonumber
&\klg\,(1+\ve)\,\big\|\,\PA\,\phi_\el^\mathrm{B}(\gamma_\ve)
\otimes\Omega\,\big\|^2+
(1+\tfrac{1}{\ve})\,\big\|\,(\PA-\PO)\,\phi_\el^\mathrm{B}(\gamma_\ve)
\otimes\Omega\,\big\|^2
\\
&\klg\,
(1+\ve)\big\|\,\PA\,\phi_\el^\mathrm{B}(\gamma_\ve)\otimes\Omega\,\big\|^2\,+\,
(1+\tfrac{1}{\ve})\,\bigO(d_*^2)\,E\,\label{tina2}
\|\phi_\el^\mathrm{B}(\gamma_\ve)\otimes\Omega\|^2\,.
\end{align}
We may assume that
$\SPb{\phi_\el^\mathrm{B}(\gamma_\ve)\otimes\Omega}{
\NPO{\gamma}\,\phi_\el^\mathrm{B}(\gamma_\ve)\otimes\Omega}$
is positive. (For otherwise the upper bound on
$\inf\spec(\NPO{\gamma})$ holds true trivially.)
Choosing $\ve= d_*$, $\ve'= d_1$,
and using \eqref{minmax-B},
\eqref{tina1}, and \eqref{tina2}, we find some $C_\gamma\in(0,\infty)$
such that
\begin{eqnarray*}
\lefteqn{
\frac{\SPb{\phi_\el^{\mathrm{B}}(\gamma_\ve)\otimes\Omega}{\PA\,(
\DA-\gamma/|\V{x}|+\Hf)\,\PA\phi_\el^{\mathrm{B}}(\gamma_\ve)\otimes\Omega}}{
\big\|\,\PA\,\phi_\el^{\mathrm{B}}(\gamma_\ve)\otimes\Omega\,\big\|^2}
}
\\
&\klg&
\frac{(1+\ve)\,E_\el(\gamma_\ve)\,
+\,
\bigO(d_1+d_*)
}{(1-\bigO(d_*))/(1+\ve)}
\,\klg\,
E_\el(\gamma)\,+\,C_\gamma\,(d_1+d_*)
\,,
\end{eqnarray*}
for all sufficiently small values of $d_*$.
Repeating the same argument with $\gamma=0$ and with
$\phi_\el^\mathrm{B}(\gamma_\ve)$ replaced by some normalized
$\chi\in C_0^\infty(\RR^3,\CC^4)$ with
$\SPn{\chi}{\PO\,\DO\,\PO\,\chi}\klg1+\ve$,
we obtain the estimate
$$
1\,\klg\,\Thnp\,\klg\,1+\bigO(d_1+d_*)\,.
$$
Since the lower bound on $\inf\spec(\NPO{\gamma})$ 
follows from Theorem~\ref{thm-sb-np} with $\rho= d_*$
 this concludes the proof.
\end{proof}


\bigskip

\noindent
{\bf Acknowledgement:} 
This work has been partially supported
by the DFG (SFB/TR12).



\end{document}